\documentclass[12pt,draftclsnofoot, onecolumn]{IEEEtran}
\usepackage{soul}
\usepackage{graphicx}
\usepackage{amsmath}
\usepackage{amsfonts}
\usepackage{array}
\newcommand{\PreserveBackslash}[1]{\let\temp=\\#1\let\\=\temp}
\newcolumntype{C}[1]{>{\PreserveBackslash\centering}p{#1}}
\newcolumntype{R}[1]{>{\PreserveBackslash\raggedleft}p{#1}}
\newcolumntype{L}[1]{>{\PreserveBackslash\raggedright}p{#1}}
\usepackage[usenames]{color}
\usepackage{colortbl,booktabs}
\usepackage{tabularx}
\usepackage{multicol}
\usepackage{booktabs}
\usepackage[Symbol]{upgreek}
\usepackage{subfigure}
\usepackage{bm}
\usepackage{url}
\usepackage[usenames,dvipsnames,svgnames,table]{xcolor}

\usepackage{cite}
\usepackage{balance}
\usepackage{times}
\newtheorem{lemma}{Lemma}
\newtheorem{proof}{Proof}

\newtheorem{remark}{Remark}

\usepackage{stfloats}
\usepackage{balance}

\usepackage[nonumberlist,acronym,shortcuts]{glossaries}
\makeglossaries
\makeindex
\begin{document}

\title{\LARGE{Adaptive Coding and Modulation for Large-Scale Antenna Array Based
 Aeronautical Communications in the Presence of Co-channel Interference}}
\author{Jiankang~Zhang,~\IEEEmembership{Member,~IEEE,}
 Sheng~Chen,~\IEEEmembership{Fellow,~IEEE,}
 Robert~G.~Maunder,~\IEEEmembership{Senior~Member,~IEEE,}
 Rong~Zhang,~\IEEEmembership{Senior~Member,~IEEE,}
 and~Lajos~Hanzo,~\IEEEmembership{Fellow,~IEEE}

\thanks{The authors are with Electronics and Computer Science, University of
 Southampton, U.K. (E-mails: jz09v@ecs.soton.ac.uk, sqc@ecs.soton.ac.uk,
 rm@ecs.soton.ac.uk, rz@ecs.soton.ac.uk, lh@ecs.soton.ac.uk) J. Zhang is
also with School of Information Engineering, Zhengzhou University, China.
S. Chen is also with King Abdulaziz University, Jeddah, Saudi Arabia.} %
\thanks{The financial support of the European Research Council's Advanced Fellow Grant, 
 the Royal Society Wolfson Research Merit Award, and the Innovate UK funded Harmonised
 Antennas, Radios, and Networks (HARNet) are gratefully acknowledged.} %
\vspace*{-5mm}
}

\maketitle

\IEEEpeerreviewmaketitle

\begin{abstract}
 In order to meet the demands of `Internet above the clouds', we propose a
 multiple-antenna aided adaptive coding and modulation (ACM) for aeronautical
 communications. The proposed ACM scheme switches its coding and modulation mode
 according to the distance between the communicating aircraft, which is readily
 available with the aid of the airborne radar or the global positioning system. We
 derive an asymptotic closed-form expression of the signal-to-interference-plus-noise
 ratio (SINR) as the number of transmitting antennas tends to infinity, in the
 presence of realistic co-channel interference and channel estimation errors. The
 achievable transmission rates and the corresponding mode-switching distance-thresholds
 are readily obtained based on this closed-form SINR formula. Monte-Carlo simulation
 results are used to validate our theoretical analysis. For the specific example
 of 32 transmit antennas and 4 receive antennas communicating at a 5 GHz carrier
 frequency and using 6 MHz bandwidth, which are reused by multiple other pairs of
 communicating aircraft, the proposed distance-based ACM is capable of providing
 as high as 65.928 Mbps data rate when the communication distance is less than 25\,km.
\end{abstract}

\linespread{2.0}
\section{Introduction}\label{S1}

 The appealing {service} of the `Internet above the clouds' \cite{jahn2003evolution} 
 motivates researchers to develop high data rate and high spectral-efficiency (SE)
 aeronautical communication techniques. Traditionally, satellite-based access has
 been the main solution for aeronautical communication. However, it suffers from
 the drawbacks of low throughput and high processing delay as well as high charges
 by the satellite providers. The aeronautical ad hoc network (AANET)
 \cite{vey2014aeronautical} concept was conceived for supporting direct communication
 and data relaying among aircraft for airborne Internet access. However, the current
 existing transmission techniques are incapable of providing the high throughput and 
 high SE communications among aircraft required by this airborne Internet access
 application.

 The planed future aeronautical communication systems, specifically, the L-band digital
 aeronautical communications system (L-DACS) \cite{schnell2014ldacs,jain2011analysis}
 and the aeronautical mobile airport communication system (AeroMACS)
 \cite{budinger2011aeronautical,bartoli2013aeromacs}, only provide upto 1.37\,Mbps
 and 9.2\,Mbps air-to-ground communication date rates, respectively. Moreover, the
 L-DACS1 air-to-air mode \cite{graupl2011ldacs1} is only capable of providing 273\,kbps
 net user rate for direct aircraft-to-aircraft communication, which cannot meet the
 high throughput demand of the Internet above the clouds. Furthermore, these rates
 are achievable for point-to-point transmissions, but multiple frequency resources are
 required for supporting multiple pairs of aircraft communications. Therefore, these
 future aeronautical communication techniques fail to satisfy the demanding requirements
 of airborne Internet access. Additionally, the L-DACS1 air-to-air mode has to collect
 and distribute the associated channel state information (CSI) to all aircraft within
 the communication range \cite{graupl2011ldacs1}, which is challenging in practical
 implementation. { Even if the air-to-air communication capacity of
 these future aeronautical communication systems could be made sufficiently high, they
 would still be forbidden for commercial airborne Internet access, because their
 frequency bands are within the bands assigned to the safety-critical air traffic control
 and management systems.}

 In order to meet the high throughput and high SE demands of the future AANET, we propose
 a large-scale antenna array aided adaptive coding and modulation (ACM) based solution
 for aeronautical communication. Before reviewing the family of ACM and multiple-antenna
 techniques, we first elaborate on the specific choice of the frequency band suitable for
 the envisaged AANET. Existing air traffic systems mainly use the  very high frequency
 (VHF) band spanning from 118\,MHz to 137\,MHz \cite{Haind2007Anindependent}, and there
 are no substantial idle frequency bands. The ultra high frequency (UHF) band has almost
 been fully occupied by television broadcasting, cell phones and satellite communications,
 including the global positioning system (GPS). Thus, no substantial idle frequency bands
 can be found in the UHF band either. This motivates us to explore the super high frequency
 (SHF) band spanning from 3\,GHz to 30\,GHz, for example, using 5\,GHz carrier frequency
 for this aeronautical communication application. { Note that even if there were sufficient unused frequency slots in the VHF and UHF bands, it is advisable not to use them
 because the frequency band of the envisaged airborne Internet access system should be
  sufficiently far away from the  bands assigned to the safety-critical air control and
 management systems, satellite communication and GPS systems.}

 ACM \cite{goldsmith1998adaptive,hanzo2002adaptive} has been demonstrated to be a
 powerful technique of increasing data rate and improving SE over wireless fading
 channels. It has been extensively investigated also in the context of IEEE 802.11
 \cite{tan2008link}, LTE-advance 4G mobile systems \cite{dahlman20134g,meng2014constellation}
 and broadband satellite communication systems \cite{alberty2007adaptive}. The
 optimal ACM relies on the perfect knowledge of the instantaneous CSI, but channel
 estimation errors are unavoidable in practical communication systems
 \cite{zhou2004adaptive}. Furthermore, the CSI of frequency division duplexing based
 systems must be obtained through a feedback channel, which potentially introduces
 feedback errors and delays \cite{zhou2004accurate}. These factors significantly degrade
 the ACM performance. In order to reduce the sensitivity to CSI errors, Zhou {\it et al.}
 \cite{zhou2004adaptive} proposed an adaptive modulation scheme relying on partial CSI,
 while Taki {\it et al.} \cite{taki2014adaptive} designed an ACM scheme based directly
 on imperfect CSI. A whole range of differentially encoded and non-coherently detected
 star-QAM  schemes were characterized in \cite{hanzo2004quadrature}, while the channel
 coding aspects were documented in \cite{hanzo2002adaptive}. Most existing research
 on ACM focused on terrestrial wireless communications, where the channels exhibit
 Rayleigh fading characteristics. But the research community seldom considered the
 propagation characteristics of aeronautical communications in designing ACM schemes. 

 In aeronautical communication, typically there is line-of-sight (LOS) propagation, in
 addition to multipath fading, where the LOS component dominates the other multipath
 components of the channel. The investigations of
 \cite{haas2002aeronautical,gligorevic2013airport} have revealed that the aeronautical
 channel can be modeled as a Rician channel for the flight phases of taxiing, landing,
 takeoff and en-route, while the aeronautical channel during the aircraft's parking
 phase can be modeled as a Rayleigh channel, which can be viewed as a specific case of
 the Rician channel with a zero Rician $K$-factor. Furthermore, multiple-antenna aided
 techniques have been employed in aeronautical communication for increasing the
 transmission capacity \cite{zhang2015interpolated}. Although it is challenging to deploy
 multiple antennas, on an aircraft \cite{zhang2017survey}, especially a large-scale
 antenna array, the development of conformal antenna \cite{gao2015conformal} has paved
 the way for deploying large-scale antenna arrays on aircraft. At the time of writing,
 however, there is a paucity of information on how much capacity can be offered by
 employing multiple antennas in aeronautical communications.

 Against this background, we develop { an ACM based and large-scale antenna array aided
 physical-layer transmission technique capable of facilitating en-route airborne Internet
 access for the future AANET employing the time division duplexing (TDD) protocol, which
 has already been adopted by the so-called automatic dependent surveillance-broadcast standard\cite{radio2002minimum}, as well as by the L-DACS and
 AeroMACS arrangements. Our main contributions are}
\begin{enumerate}
\item We propose and analyze a distance-information based ACM scheme for large-scale
 antenna array aided aeronautical communication in SHF band, which switches its ACM
 mode based on the distance between the desired pair of communicating aircraft. This
 scheme is more practical than the existing ACM schemes that rely on instantaneous
 channel-quality metrics, such as the instantaneous signal-to-noise ratio (SNR). This
 is because in aeronautical communication it is extremely difficult to acquire an
 accurate estimate for any instantaneous channel-quality metric, and an ACM based on
 such a switching metric will frequently fail. By contrast, the accurate distance
 information between the communicating aircraft can readily be acquired with the aid
 airborne radar. Alternatively, the accurate position information can also be acquired
 with the assistance of GPS.
\item We explicitly derive a closed-form expression of the asymptotic
 signal-to-interference-plus-noise ratio (SINR) of multiple-antenna aided aeronautical
 communication in the presence of realistic channel estimation errors and the co-channel
 interference imposed by other aircraft within the communication range. This closed-form
 SINR formula enables us to directly derive the achievable theoretical transmission rates
 and the associated mode-switching distance-thresholds for the proposed ACM scheme. 
 Moreover, as a benefit of large-scale antenna arrays, every pair of communicating
 aircraft in our system uses the same frequency resource block, which dramatically
 enhances the system's SE.
\end{enumerate}

 The rest of this paper is organized as follows. Section~\ref{S2} describes the
 multiple-antenna aided aeronautical communication system, with emphasis on the
 propagation and signal models. Section~\ref{S3} is devoted to the proposed
 distance-based ACM scheme, including the derivation of the closed-form  asymptotic
 SINR in the presence of imperfect CSI and co-channel interference that leads to our
 detailed design of the achievable data rates and the associated switching distance
 thresholds. Section~\ref{S4} presents our simulation results for characterizing the
 impact of the relevant parameters in aeronautical communication. Our conclusions are
 given in Section~\ref{S5}.

\section{System model}\label{S2}

\begin{figure*}[bp!]
\vspace*{-4mm}
\begin{center}
 \includegraphics[width=0.95\textwidth]{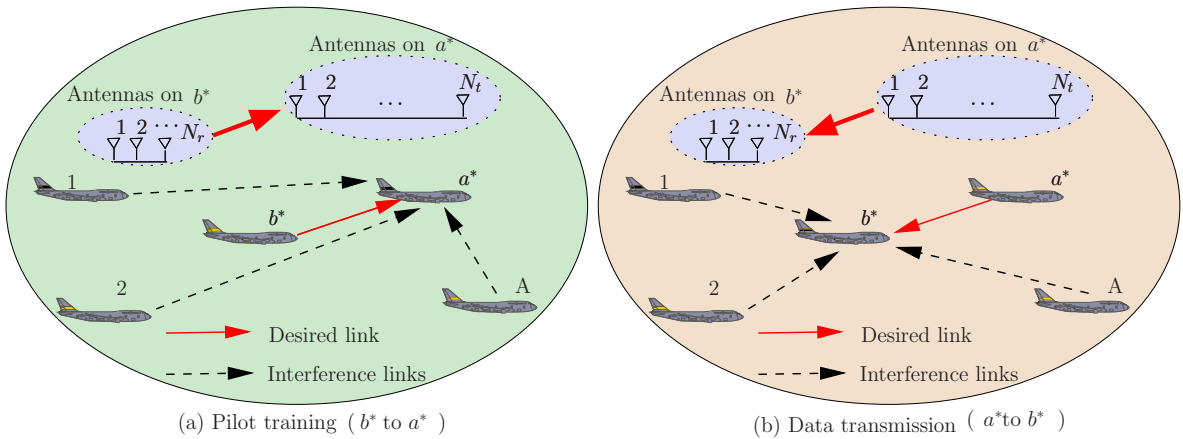}
\end{center}
\vspace*{-5mm}
\caption{A TDD based aeronautical communication model with co-channel interference,
 where aircraft $a^*$ is transmitting data to aircraft $b^*$.}
\label{FIG1}
\vspace*{-2mm}
\end{figure*}

 Fig.~\ref{FIG1} depicts an aeronautical communication system consisting of $(A + 2)$
 aircraft { at their cruising altitudes} within a given communication zone{\footnote{{ Since
 Internet access is forbidden at takeoff and landing, it is reasonable to consider
 only aircraft en-route at cruising altitude.}}}. Aircraft $a^{*}$ and $b^{*}$ are the
 desired pair of communicating aircraft, with $a^{*}$ transmitting data to $b^{*}$,
 while aircraft $a = 1, 2, \cdots, A$ are interfering aircraft. As the system operates
 in SHF band, it is feasible to design compact high-gain antennas
 \cite{parsons2000themobile} and hence to deploy a large-scale antenna array on an
 aircraft. This also enables every pair of communicating aircraft to use the same
 frequency and time slot. Specifically, the system is based on the TDD protocol, and
 each aircraft has $N_{\text{total}}$ antennas, which are capable of transmitting and
 receiving on the same frequency. Furthermore, from these $N_{\text{total}}$ antennas,
 $N_t$ antennas are utilized for transmitting data, while $N_r$ antennas are utilized
 for receiving data. We assume that the number of data-receiving antennas (DRAs) is no
 higher than that of the data-transmitting antennas (DTAs), i.e., {{$N_r \le N_t < 
 N_{\text{total}}$}}. The reasons are: 1)~The spatial degrees of freedom
 $\text{min}\left\{N_r, N_{t}\right\}$ determine the supportable data streams, and
 thus the number of DRAs should be no higher than the number of DTAs in order to make
 sure that the number of data streams after transmit precoding is no higher than
 $\text{min}\left\{N_r, N_{t}\right\}$; and 2)~The remaining $\left(N_{\rm total} -
 N_r\right)$ or $\left(N_{\text{total}} - N_t\right)$ antennas are capable of
 transmitting/receiving other information, such as air traffic control or emergency
 information, at a frequency different from that of the data transmission. The system
 adopts orthogonal frequency-division multiplexing (OFDM) for mitigating the multipath
 effects and for achieving a high SE. We further assume that all the aircraft are jumbo
 jets, and they are all equipped with a same large-scale antenna array. We also assume
 that each aircraft has an airborne radar capable of measuring the distance to nearby
 aircraft. Alternatively, the distance information may be acquired with the aid of GPS.

 As illustrated  in Fig.~\ref{FIG1}\,(b), aircraft $a^{*}$ employs $N_t$ antennas to
 transmit data and aircraft $b^{*}$ employs $N_r$ antennas to receive data. Therefore,
 $a^{*}$ has to know the multi-input multi-output (MIMO) channel matrix linking its
 $N_t$ DTAs to the $N_r$ DRAs of $b^{*}$ in order to carry out transmit precoding. The 
 channel estimation is performed by pilot based training before data transmission, as
 shown in Fig.~\ref{FIG1}\,(a). More specifically, aircraft $b^{*}$ transmits pilot
 symbols from its $N_r$ DRAs to aircraft $a^{*}$'s $N_t$ DTAs to enable $a^{*}$ to
 perform the training based channel estimation. The required MIMO channel matrix can
 then be obtained by exploiting the TDD channel's reciprocity.

\subsection{Pilot based training}\label{S2.1}

 During the pilot based training stage shown in Fig.~\ref{FIG1}\,(a), aircraft $b^{*}$
 transmits its pilot symbols via its $N_{r}$ DRAs to the $N_{t}$ DTAs of
 aircraft $a^{*}$, and this pilot based training is interfered by all adjacent aircraft
 within the communication range. The most serious interference is imposed when the
 interfering aircraft $a = 1, 2, \cdots, A$ also transmit the same pilot symbols as
 aircraft $b^{*}$. The frequency domain (FD) MIMO channel between the $N_r$ DRAs of
 $b^{*}$ and the $N_t$ DTAs of $a^{*}$, on the $n$th OFDM subcarrier of the $s$th symbol,
 is defined by the matrix $\bm{H}_{a^{*}}^{b^{*}}[s,n] \in \mathbb{C}^{N_t\times N_r}$,
 where $0\le n\le N-1$ with $N$ denoting the length of the OFDM block. Let us assume
 that the length of the cyclic prefix (CP) $N_{\text{cp}}$ is longer than the channel
 impulse response (CIR) $P$. The FD discrete signal model at the $n$th subcarrier and
 $s$th time interval during the pilot training period can be written as\setcounter{equation}{0}
\begin{align}\label{eq1}
 & \widetilde{\bm{Y}}_{a^{*}}[s,n] = \sqrt{P_{r,a^{*}}^{b^{*}}} \bm{H}_{a^{*}}^{b^{*}}[s,n]
  \widetilde{\bm{X}}^{b^{*}}[s,n]  
  + \sum\limits_{a=1}^A \sqrt{P_{r,a^{*}}^a} \bm{H}_{a^{*}}^a[s,n]
  \widetilde{\bm{X}}^a[s,n] + \widetilde{\bm{W}}_{a^{*}}[s,n] ,
\end{align}
 where $\widetilde{\bm{Y}}_{a^{*}}[s,n]\! =\!\left[\widetilde{Y}^{a^*}_1[s,n] ~
 \widetilde{Y}^{a^*}_2[s,n]\cdots \widetilde{Y}^{a^*}_{N_t}[s,n]\right]^{\rm T}\!
 \in \mathbb{C}^{N_t\times 1}$ represents the signal vector received by the $N_t$ DTAs of
 $a^{*}$, $\widetilde{\bm{X}}^{b^{*}}[s,n]\! =\! \left[\widetilde{X}^{b^{*}}_1[s,n] ~
 \widetilde{X}^{b^{*}}_2[s,n]\cdots \widetilde{X}^{b^{*}}_{N_r}[s,n]\right]^{\rm T}\!
 \in \mathbb{C}^{N_r\times 1}$ is the transmitted pilot symbol vector, which obeys the
 complex white Gaussian distribution with a zero mean vector and the covariance matrix
 of $\bm{I}_{N_r}$, namely, $\widetilde{\bm{X}}^{b^{*}}[s,n]\sim
 \mathcal{CN}\left(\bm{0},\bm{I}_{N_r}\right)$, and $\widetilde{\bm{W}}_{a^{*}}[s,n]\sim
 \mathcal{CN}\left(\bm{0},\sigma_w^2\bm{I}_{N_t}\right)$ is the associated FD additive
 white Gaussian noise (AWGN) vector, while the worst-case co-channel interference is
 considered with $\widetilde{\bm{X}}^{a}[s,n]=\widetilde{\bm{X}}^{b^{*}}[s,n]$ for
 $1\le a\le A$. Here we have used $\bm{0}$ to represent the zero vector of an appropriate
 dimension and $\bm{I}_{N_r}$ to represent the $(N_r\times N_r)$-element identity matrix.
 Furthermore, in~(\ref{eq1}), $P_{r,a^{*}}^{b^{*}}$ and $P_{r,a^{*}}^a$ denote the
 received powers of the desired signal and the interference signal transmitted from
 $b^*$ and $a$ { at a single receive antenna}, respectively.

 The FD channel transfer function coefficient (FD-CHTFC) matrix $\bm{H}_{a^{*}}^{b^{*}}[s,n]$ 
 is explicitly given by
\begin{align}\label{eq2}
 \bm{H}_{a^{*}}^{b^{*}}[s,n] =& \nu \bm{H}_{\text{d},a^{*}}^{b^{*}}[s,n] +
  \varsigma \bm{H}_{\text{r},a^{*}}^{b^{*}}[s,n] ,
\end{align}
 where $\bm{H}_{\text{d},a^{*}}^{b^{*}}[s,n]\in \mathbb{C}^{N_t\times N_r}$ and
 $\bm{H}_{\text{r},a^{*}}^{b^{*}}[s,n]\in \mathbb{C}^{N_t\times N_r}$ denote the
 deterministic and scattered channel components, respectively, while $\nu =
 \sqrt{\frac{K_{\text{Rice}}}{K_{\text{Rice}}+1}}$ and $\varsigma =
 \sqrt{\frac{1}{K_{\text{Rice}} + 1}}$, with $K_{\text{Rice}}$ being the $K$-factor
 of the Rician channel. Furthermore, the scattered channel component is given by \cite{kim2010spatial}
\begin{align}\label{eq3}
 \bm{H}_{\text{r},a^{*}}^{b^{*}}[s,n] =& \bm{R}_{a^{*}}^{\frac{1}{2}}
 \bm{G}_{a^{*}}^{b^{*}}[s,n] \bm{R}^{b^{*}\frac{1}{2}} ,
\end{align}
 where $\bm{R}_{a^*}\in \mathbb{C}^{N_t\times N_t}$ and $\bm{R}^{b^*}\in
 \mathbb{C}^{N_r\times N_r}$ are the spatial correlation matrices for the $N_t$
 antennas of $a^*$ and the $N_r$ antennas of $b^*$, respectively, while
 $\bm{G}_{a^{*}}^{b^{*}}[s,n]\in \mathbb{C}^{N_t\times N_r}$ has the independently
 identically distributed complex-valued entries, each obeying  the distribution
 $\mathcal{CN}(0,1)$. As a benefit of the CP, the OFDM symbols do not overlap in time
 and the processing can be carried out on a per-carrier basis \cite{zhang2011joint}.
 Hence, to simplify our notations, we will omit the OFDM symbol index $s$ and the
 subcarrier index $n$ in the sequel.

 Let $\bm{vec}(\bm{A})$ denote the column stacking operation applied to matrix $\bm{A}$.
 Clearly, \\ $\mathcal{E}\left\{\bm{vec}\big(\bm{H}_{\text{r},a^{*}}^{b^{*}}\big)\right\}
 =\bm{0}$, where $\mathcal{E}\{ ~\}$ denotes the expectation operator, and the
 covariance matrix $\bm{R}_{\text{r},a^{*}}^{b^*}\in\mathbb{C}^{N_{t}N_{r}\times N_{t}N_{r}}$
 of $\bm{vec}\big(\bm{H}_{\text{r},a^{*}}^{b^{*}}\big)$ is given by\setcounter{equation}{3}
\begin{align}\label{eq4}
 \bm{R}_{\text{r},a^{*}}^{b^*} \! =& \mathcal{E}\! \left\{\! \bm{vec}\left(\bm{H}_{\text{r},a^{*}}^{b^*}\right)
  \! \left(\bm{vec}\left(\bm{H}_{\text{r},a^{*}}^{b^*}\right)\right)^{\rm H}\right\} 
  \! =\!  \bm{R}^{b^*}\! \otimes\! \bm{R}_{a^{*}} \! , \!
\end{align}
 where $\otimes$ denotes the matrix Kronecker product. Therefore,
 $\bm{vec}\left(\bm{H}_{a^{*}}^{b^{*}}\right)\sim \mathcal{CN}\left(
 \bm{vec}\left(\nu\bm{H}_{\text{d},a^{*}}^{b^*}\right),\right.$ $\left.\bm{R}_{\text{r},a^{*}}^{b^*}\right)$.
 Since all the aircraft have the same antenna array, we make the assumption that all
 the spatial correlation matrices $\bm{R}_{a_t}$ for $a_t\in\mathcal{A}=\{1,2,\cdots ,
 A,a^*,b^*\}$ are approximately equal, i.e., we assume that $\bm{R}_{a_t}=\bar{\bm{R}}_t$
 $\forall a_t\in\mathcal{A}$ hold{ \footnote{{ It is reasonable to assume that all 
 jumbo jets are equipped with identical antenna array. However, because the geometric shapes
 of different types of jumbo jets are slightly different, $\bm{R}_{a_t}=\bar{\bm{R}}_t$
 $\forall a_t\in\mathcal{A}$ only hold approximately.}}}.
 By the same argument, we assume that $\bm{R}^{a_r}=\bar{\bm{R}}^r$ $\forall a_r\in
 \mathcal{A}$ hold{ \footnote{{ Similarly, $\bm{R}^{a_r}=\bar{\bm{R}}^r$
 $\forall a_r\in \mathcal{A}$ only hold approximately.}}}.
 Thus, all the channels' covariance matrices are assumed to be equal, i.e., we have
\begin{align}\label{eq5}
 \bm{R}_{\text{r},a_t}^{a_r} =& \bar{\bm{R}}_{\text{r},t}^{r} = \bar{\bm{R}}^r \otimes
  \bar{\bm{R}}_t, \, \forall a_t,a_r\in \mathcal{A} \text{ and } a_t\neq a_r .
\end{align}
 This implies that every aircraft has the knowledge of its channel covariance matrix.
 For example, aircraft $a^*$ knows its antenna array's spatial correlation matrices
 $\bm{R}_{a^*}=\bar{\bm{R}}_t$ and $\bm{R}^{a^*}=\bar{\bm{R}}^r$, and since $\bm{R}^{b^*}
 =\bm{R}^{a^*}$, it knows its channel covariance matrix $\bm{R}_{\text{r},a^{*}}^{b^*}=
  \bm{R}^{b^*}\otimes \bm{R}_{a^{*}}=\bm{R}^{a^*}\otimes \bm{R}_{a^{*}}=\bar{\bm{R}}^r
 \otimes\bar{\bm{R}}_t$. It can be seen that $\bm{R}^{b^*}=\bm{R}^{a^*}$ is the real
 assumption required{ \footnote{{ Alternatively, we can also avoid imposing this
 assumption. Then, $a^*$ can ask $b^*$ to send its antenna correlation matrix
 $\bm{R}^{b^*}$. For example, during the initial handshake of establishing the
 link, $b^*$ can sends $\bm{R}^{b^*}$ to $a^*$ through the signaling at the expense of
 increasing the signaling overhead.}}}.

 The received power $P_{r,a^{*}}^{b^{*}}$ at aircraft $a^{*}$ is linked to the transmitted
 signal power $P_t^{b^{*}}$ of aircraft $b^*$ by the following path loss model
 \cite{parsons2000themobile}
\begin{equation}\label{eq6}
 P_{r,a^{*}}^{b^{*}} = P_{t}^{b^{*}}10^{-0.1 L_{\text{path loss},a^{*}}^{b^{*}}} ,
\end{equation}
 where $L_{\text{path loss},a^{*}}^{b^{*}}$ represents the path loss in dB, which
 can be modeled as \cite{parsons2000themobile}
\begin{equation}\label{eq7}
 L_{\text{path loss},a^{*}}^{b^{*}} \,[\text{dB}] = -154.06 + 20\log_{10}\left(f\right) + 20\log_{10}\left(d\right) ,
\end{equation}
 where $f$ [Hz] is the carrier frequency and $d$ [m] is the distance between the
 transmit antenna and the receive antenna. For the received interference signal
 power $P_{r,a^{*}}^a$, we have a similar path loss model.
 Each entry of the FD AWGN vector $\bm{W}_{a^{*}}$ obeys the distribution
 ${\cal CN}(0,\sigma_w^2)$ with $\sigma_w^2=\frac{P_N}{N}$, in which $P_N$ is the receiver
 noise power given by \cite{sheriff2003mobile}
\begin{equation}\label{eq8}
 P_N = F k T_0 B ,
\end{equation}
 where $F$ [dB] is the receiver's noise figure, $T_0$ is the reference temperature in
  Kelvin at the receiver, $k=1.3 \times 10^{-23}$ is Boltzmann's constant and
 $B$ [Hz] is the bandwidth. 
 
 Since every aircraft knows its channel covariance matrix, we can apply the optimal
 minimum mean square error (MMSE) estimator \cite{kay2003fundamentals} to estimate
 the channel matrix $\bm{H}_{a^{*}}^{b^{*}}$. The MMSE estimate
 $\widehat{\bm{H}}_{a^{*}}^{b^{*}}$ of $\bm{H}_{a^{*}}^{b^{*}}$ is given by
\begin{align}\label{eq9}
  &\bm{vec}\left(\widehat{\bm{H}}_{a^*}^{b^*}\right) = \bm{vec}\left(\nu\bm{H}_{\text{d},a^*}^{b^*}\right)
  + \varsigma^2\bar{\bm{R}}_{\text{r},t}^r \left(\frac{\sigma_w^2}{P_{r,a^*}^{b^*}} \bm{I}_{N_r N_t}
  + \varsigma^2\bar{\bm{R}}_{\text{r},t}^r
  + \sum\limits_{a=1}^A \frac{P_{r,a^*}^a}{P_{r,a^*}^{b^*}} \varsigma^2 \bar{\bm{R}}_{\text{r},t}^r \right)^{-1} 
  \nonumber \\ &\hspace*{10mm} \times \left( \bm{vec}\left(\varsigma \bm{H}_{\text{r},a^*}^{b^*}\right)
  + \sum\limits_{a=1}^A\sqrt{\frac{P_{r,a^*}^a}{P_{r,a^*}^{b^*}}}
  \bm{vec}\left(\varsigma\bm{H}_{\text{r},a^*}^a\right) 
  + \frac{1}{\sqrt{P_{r,a^*}^{b^*}}} \bm{vec}\left(\widetilde{\bar{\bm{W}}}_{a^*}
  \Big(\widetilde{\bar{\bm{X}}}^{b^*}\Big)^{\rm H} \right) \right) .
\end{align}
 where $\widetilde{\bar{\bm{X}}}^{b^*}\! \in\!
 \mathbb{C}^{N_r\times N_r}$ consists of the $N_r$ pilot symbols with
 $\widetilde{\bar{\bm{X}}}^{b^*}\big(\widetilde{\bar{\bm{X}}}^{b^*}\big)^{\rm H}
 \! =\! \bm{I}_{N_r}$, and $\widetilde{\bar{\bm{W}}}_{a^*}\! \in\! \mathbb{C}^{N_r\times N_r}$
 is the corresponding AWGN vector over the $N_{r}$ consecutive OFDM pilot symbols.
 It is well-known that the distribution of the MMSE estimator (\ref{eq9}) is
 $\bm{vec}\big(\widehat{\bm{H}}_{a^{*}}^{b^{*}}\big)\sim \mathcal{CN}\left(
 \bm{vec}\left(\nu\bm{H}_{\text{d},a^{*}}^{b^*}\right),\right.$ $\left.\bm{\Phi}_{a^{*}}^{b^{*}}\right)$ 
 \cite{kay2003fundamentals}, that is, $\bm{vec}\left(\widehat{\bm{H}}_{a^{*}}^{b^{*}}\right)$
 is an unbiased estimate of $\bm{vec}\left(\bm{H}_{a^{*}}^{b^*}\right)$ with
 the estimation accuracy specified by the covariance matrix $\bm{\Phi}_{a^{*}}^{b^{*}}\in
 \mathbb{C}^{N_{t}N_{r}\times N_{t}N_{r}}$, which is given by \setcounter{equation}{9}
\begin{align}\label{eq10}
 \bm{\Phi}_{a^*}^{b^*} =& \varsigma^2 \bar{\bm{R}}_{\text{r},t}^r \left(
  \frac{\sigma_w^2}{P_{r,a^{*}}^{b^{*}}} \bm{I}_{N_r N_t} \! +\! \varsigma^2 \bar{\bm{R}}_{\text{r},t}^r
  \! + \! \sum\limits_{a=1}^A \frac{P_{r,a^{*}}^a}{P_{r,a^{*}}^{b^{*}}} \varsigma^2
  \bar{\bm{R}}_{\text{r},t}^r\right)^{-1} 
  \varsigma^2 \bar{\bm{R}}_{\text{r},t}^r .
\end{align}

\subsection{Data transmission}\label{S2.2}

 During the data transmission, aircraft $a^*$ transmits the symbols $\bm{X}^{a^*}=\big[
 X_1^{a*} ~ X_2^{a*}\cdots X_{N_r}^{a*}\big]^{\rm T}\in\mathbb{C}^{N_r\times 1}$ from
 its $N_t$ DTAs to the $N_r$ DRAs of aircraft $b^*$. Let us denote  the MIMO channel
 matrix during this data transmission as $\bm{H}_{b^*}^{a^*}\in \mathbb{C}^{N_r\times N_t}$.
 Then, upon exploiting the channel's reciprocity in TDD systems, we have $\bm{H}_{b^*}^{a^*}
 =\left(\bm{H}_{a^*}^{b^*}\right)^{\rm H}$.

 To mitigate the interference between multiple antennas, transmit precoding (TPC) is
 adopted for data transmission. There are various methods of designing the TPC matrix
 $\bm{V}_{b^*}^{a^*}\in\mathbb{C}^{N_t\times N_r}$, including the convex
 optimization based method of \cite{gershman2010convex}, the minimum variance method of
 \cite{lorenz2005robust}, the minimum bit-error rate (MBER) design of \cite{Yao2009mberpc},
 the MMSE design of \cite{Vojcic1998mmsepc} and the zero-forcing (ZF) design as well
 as the eigen-beamforming or matched filter (MF) design of  \cite{hoydis2013massive}.
 For a large-scale MIMO system, the complexity of the optimization based TPC designs of
 \cite{gershman2010convex,lorenz2005robust,Yao2009mberpc} may become excessive.
 Additionally, in this case, the performance of the MBER design \cite{Yao2009mberpc} is
 indistinguishable from the MMSE one. Basically, for large-scale antenna array based MIMO,
 the performance of the MMSE, ZF and MF based TPC solutions are sufficiently good. The
 MF TPC design offers the  additional advantage of the lowest complexity and, therefore,
 it is chosen in this work. Specifically, the MF TPC matrix is given by \setcounter{equation}{10}
\begin{equation}\label{eq11}
 \bm{V}_{b^*}^{a^*} = \left(\widehat{\bm{H}}_{b^*}^{a^*}\right)^{\rm H} = \widehat{\bm{H}}_{a^*}^{b^*} ,
\end{equation}
 where $\widehat{\bm{H}}_{b^*}^{a^*}$ denotes the estimate of $\bm{H}_{b^*}^{a^*}$,
 and $\widehat{\bm{H}}_{a^*}^{b^*}$ is the channel estimate obtained during pilot
 training.

 In the presence of the interference imposed by aircraft $a$ for $1\le a\le A$,
 the received signal vector at aircraft $b^*$, $\bm{Y}_{b^*}=\big[Y_1^{b^*} ~ Y_2^{b^*}
 \cdots Y_{N_r}^{b^*} \big]^{\rm T}\in\mathbb{C}^{N_r\times 1}$, can be written as
\begin{align}\label{eq12}
 \bm{Y}_{b^*} =& \sqrt{P_{r,b^*}^{a^*}}\bm{H}_{b^*}^{a^*}\bm{V}_{b^*}^{a^*}
  \bm{X}^{a^*} \! +\! \sum\limits_{a=1}^A \sqrt{P_{r,b^*}^a}\bm{H}_{b^*}^a\bm{V}_{b^a}^a \bm{X}^a
  \! +\! \bm{W}_{b^*} ,
\end{align}
 where $\bm{V}_{b^a}^a\in\mathbb{C}^{N_t\times N_r}$ denotes the TPC matrix at
 aircraft $a$ transmitting the signal $\bm{X}^a=\big[X_1^a ~ X_2^a \cdots
 X_{N_r}^a\big]^{\rm T}$ to its desired receiving aircraft $b^a$ for $b^a\ne b^*$, and 
 the channel's AWGN vector is $\bm{W}_{b^*}=\big[W_1^{b^*} ~ W_2^{b^*} \cdots
 W_{N_r}^{b^*}\big]^{\rm T}\sim \mathcal{CN}\big(\bm{0},\sigma_w^2\bm{I}_{N_r}\big)$.
 In particular, the signal received at the antenna $n_r^{*}$ of aircraft $b^*$, for
 $1\le n_r^*\le N_r$, is given by 
\begin{align}\label{eq13}
 Y^{b^*}_{n_r^{*}} =& \sqrt{P_{r,b^*}^{a^*}} \left[\bm{H}_{b^*}^{a^*}\right]_{[n_r^{*}:~]}
  \bm{V}_{b^*}^{a^*} \bm{X}^{a^*} + \sum\limits_{a=1}^A \sqrt{P_{r,b^*}^a}
  \left[\bm{H}_{b^*}^a\right]_{[n_r^{*}:~]} \bm{V}_{b^a}^a \bm{X}^a + W^{b^*}_{n_r^{*}} \nonumber \\ &
  = \sqrt{P_{r,b^*}^{a^*}} \left[\bm{H}_{b^*}^{a^*}\right]_{[n_r^{*}:~]}
  \left[\bm{V}_{b^*}^{a^*}\right]_{[~:n_r^{*}]} X_{n_r^{*}}^{a^*} 
  + \sum\limits_{n_r \neq n_r^{*}}\sqrt{P_{r,b^*}^{a^*}} \left[\bm{H}_{b^*}^{a^*}\right]_{[n_r^{*}:~]}
  \left[\bm{V}_{b^*}^{a^*}\right]_{[~:n_r]} X_{n_r}^{a^*} \nonumber \\ &
  + \sum\limits_{a=1}^A\sum\limits_{n_r = 1}^{N_{r}}\sqrt{P_{r,b^*}^a}\left[\bm{H}_{b^*}^a\right]_{[n_r^{*}:~]}
  \left[\bm{V}_{b^a}^a\right]_{[~:n_r]} X_{n_r}^{a} + W^{b^*}_{n_r^{*}} ,  
\end{align}
 where
 $[\bm{A}]_{[n_r:~]}\in\mathbb{C}^{1\times M}$ denotes the $n_r$th row of
 $\bm{A}\in\mathbb{C}^{N\times M}$ and $[\bm{A}]_{[~:n_r]}\in\mathbb{C}^{N\times 1}$
 denotes the $n_r$th column of $\bm{A}$.

\section{Proposed adaptive coding and modulation}\label{S3}

 In order to attain the required high-throughput transmission for our large-scale antenna
 assisted aeronautical communication system, we propose a distance-based ACM scheme. We
 begin by analyzing the system's achievable throughput, followed by the detailed design
 of this distance-based ACM.

\subsection{Achievable throughput analysis}\label{S3.1}

 The transmitting aircraft $a^*$ pre-codes its signals based on the channel estimate
 obtained during the pilot training by exploiting the TDD channel's reciprocity.
 However, the receiving aircraft $b^*$ does not have this CSI. Thus, the ergodic
 achievable rate is adopted for analyzing the achievable throughput. In order to explicitly
 derive this capacity, we rewrite the received signal at the antenna $n_r^*$ of aircraft
 $b^*$, namely, $Y^{b^*}_{n_r^{*}}$ of (\ref{eq13}), in the form given in (\ref{eq14}).

\begin{align}\label{eq14}
 Y^{b^*}_{n_r^*}\! =& \mathcal{E}\! \left\{\! \sqrt{P_{r,b^*}^{a^*}} \left[\bm{H}_{b^*}^{a^*}\right]_{[n_r^{*}:~]}
  \!\! \left[\bm{V}_{b^*}^{a^*}\right]_{[~:n_r^{*}]}\right\} \! X_{n_r^{*}}^{a^*} \!
  + \! \left(\! \sqrt{P_{r,b^*}^{a^*}} \left[\bm{H}_{b^*}^{a^*}\right]_{[n_r^{*}:~]} \!\!
  \left[\bm{V}_{b^*}^{a^*}\right]_{[~:n_r^{*}]} \right. \!\! \nonumber \\
 &-\! \left.\mathcal{E}\! \left\{\! \sqrt{P_{r,b^*}^{a^*}}
  \left[\bm{H}_{b^*}^{a^*}\right]_{[n_r^{*}:~]} \!\!
  \left[\bm{V}_{b^*}^{a^*}\right]_{[~:n_r^{*}]}\right\}\! \right)\! X_{n_r^{*}}^{a^*} + \sum\limits_{n_r \neq n_r^{*}}\sqrt{P_{r,b^*}^{a^*}} \left[\bm{H}_{b^*}^{a^*}\right]_{[n_r^{*}:~]}
  \left[\bm{V}_{b^*}^{a^*}\right]_{[~:n_r]} X_{n_r}^{a^*}  \nonumber \\
 &  +  \sum\limits_{a=1}^A\sum\limits_{n_r = 1}^{N_{r}} 
  \sqrt{P_{r,b^*}^a}\left[\bm{H}_{b^*}^a\right]_{[n_r^{*}:~]}
  \left[\bm{V}_{b^a}^a\right]_{[~:n_r]} X_{n_r}^{a} \! +\!  W^{b^*}_{n_r^{*}}  . 
\end{align}
 { Observe that the first term of (\ref{eq14}) is the desired signal, the second term is the 
 interference caused by the channel estimation error, the third term is the interference
 arriving from the other antennas of $b^*$, and the fourth term is the interference impinging from the 
 interfering aircraft, while the last term is of course the noise.} Thus, the SINR
 at the $n_r$-th antenna of $b^*$, denoted by $\gamma_{b^{*},n_{r}}^{a^{*}}$, is given by\setcounter{equation}{14}
\begin{align}\label{eq15}
 \gamma_{b^{*},n_{r}}^{a^{*}} =& \frac{P_{r,b^*}^{a^*} \left|\mathcal{E}\left\{
  \left[\bm{H}_{b^*}^{a^*}\right]_{[n_r^{*}:~]}
 \left[\bm{V}_{b^*}^{a^*}\right]_{[~:n_r^{*}]}\right\}\right|^2}
  {P_{{\text{I\&N}}_{b^*,n_r}^{a^*}}} ,
\end{align}
 in which the power of the interference pluse noise is
\begin{align}\label{eq16}
 & P_{{\text{I\&N}}_{b^*,n_r}^{a^*}} = \sigma_w^2 + P_{r,b^*}^{a^*} \text{Var}\left\{
  \left[\bm{H}_{b^*}^{a^*}\right]_{[n_r^{*}:~]}
  \left[\bm{V}_{b^*}^{a^*}\right]_{[~:n_r^{*}]}\right\} 
 + P_{r,b^*}^{a^*} \sum\limits_{n_r\neq n_r^*}\mathcal{E}\left\{\left| 
  \left[\bm{H}_{b^*}^{a^*}\right]_{[n_r^{*}:~]} 
  \left[\bm{V}_{b^*}^{a^*}\right]_{[~:n_r]}\right|^{2}\right\} \nonumber \\
 & \hspace*{10mm}+ \sum\limits_{a=1}^A P_{r,b^*}^a \sum\limits_{n_r=1}^{N_r}\mathcal{E}\left\{\left|
 \left[\bm{H}_{b^*}^a\right]_{[n_r^{*}:~]}\left[\bm{V}_{b^a}^a\right]_{[~:n_r]}\right|^{2}\right\} ,
\end{align}
 where $\text{Var}\left\{~\right\}$ denoting the variance. The achievable transmission
 rate per antenna between the transmitting aircraft $a^{*}$ and destination aircraft
 $b^{*}$ can readily be obtained as
\begin{align}\label{eq17}
 C_{b^{*}}^{a^{*}} =& \frac{1}{N_{r}}\sum\limits_{n_{r} = 1}^{N_{r}}\log_{2}\left(1
  + \gamma_{b^{*},n_{r}}^{a^{*}}\right) .
\end{align}

 As mentioned previously, the distance between the transmitting aircraft and the receiving 
 aircraft is available with the aid of airborne radar or GPS. But we do not require that
 the distances between the interfering aircraft and the desired destination aircraft are
 known to the transmitting aircraft. Realistically, the distance between two aircraft can 
 be assumed to follow the uniform distribution within the range of $\left[D_{\min}, ~
 D_{\max}\right]$, where $D_{\min}$ is the minimum separation distance required by safety 
 and $D_{\max}$ is the maximum communication range \cite{sakhaee2006global}. For example,
 we have $D_{\max}=400$\,nautical miles, which is approximately 740.8\,km, for a typical
 cruising altitude of 10.68 km. Normally, the International Civil Aviation Organization
 prescribes the minimum separation as 5 nautical miles (approximately 9.26\,km) when
 surveillance systems are in use. But the minimum separation distance could be reduced to
 2.5 nautical miles (about 4.63 \,km) when surveillance radars are intensively deployed,
 such as in an airport's airspace. Thus, the minimum distance is set to $D_{\min} = 5$\,km
 in the envisaged AANET. Intuitively, an aircraft always transmits its signal to an aircraft
 having the best propagation link with it for relying its information in the AANET. Here we
 simply assume that a pair of aircraft having the shortest communication distance have the
 best propagation link, since large-scale fading dominates the quality of propagation in
 aeronautical communication.
 Being in mind (\ref{eq6}) and (\ref{eq7}) as well as the fact that the distance $d$ is
 uniformly distributed in $\left[D_{\min}, ~ D_{\max}\right]$, we can express the average
 received signal power for the transmission from aircraft $a$ to aircraft $b^*$ as
\begin{align}\label{eq18}
 \bar{P}_r =& \bar{P}_{r,b^{*}}^{a} = \mathcal{E}\left\{P_{r,b^{*}}^{a}\right\}
  = P_{t}^{a} \frac{10^{15.406}}{f^2} \frac{1}{D_{\max}D_{\min}} .
\end{align}

 The relationship between the MMSE estimate $\left[\widehat{\bm{H}}_{b^*}^{a^*}\right]_{[n_r:~]}$
 and the true MIMO channel \\ $\left[\bm{H}_{b^*}^{a^*}\right]_{[n_r:~]}$ can be expressed as
\begin{align}\label{eq19}
 \left[\bm{H}_{b^*}^{a^*}\right]_{[n_r:~]} =& \left[\widehat{\bm{H}}_{b^*}^{a^*}\right]_{[n_r:~]}
  + \left[\widetilde{\bm{H}}_{b^*}^{a^*}\right]_{[n_r:~]} ,
\end{align}
 where $\left[\widetilde{\bm{H}}_{b^*}^{a^*}\right]_{[n_r:~]}$ denotes the estimation error,
 which is statistically independent of both \\ $\left[\widehat{\bm{H}}_{b^*}^{a^*}\right]_{[n_r:~]}$
 and $\left[\bm{H}_{b^*}^{a^*}\right]_{[n_r:~]}$ \cite{kay2003fundamentals}. Similar to the
 distribution of $\bm{vec}\left(\widehat{\bm{H}}_{a^*}^{b^*}\right)$, we have
\begin{align}\label{eq20}
  & \bm{vec}\left(\widehat{\bm{H}}_{b^{*}}^{a^{*}}\right)\sim
 \mathcal{CN}\left(\bm{vec}\left(\nu\bm{H}_{\text{d},b^{*}}^{a^*}\right),
 \bm{\Phi}_{b^{*}}^{a^{*}}\right) ,
\end{align}
 where { the covariance matrix $\bm{\Phi}_{b^{*}}^{a^{*}}$ of the MMSE estimate
 $\bm{vec}\left(\widehat{\bm{H}}_{b^{*}}^{a^{*}}\right)$} is given by.
\begin{align}\label{eq21}
  \bm{\Phi}_{b^*}^{a^*} =& \varsigma^2 \bar{\bm{R}}_{\text{r},r}^t \left(
  \frac{\sigma_w^2}{P_{r,b^{*}}^{a^{*}}} \bm{I}_{N_r N_t} \! +\! \varsigma^2 \bar{\bm{R}}_{\text{r},r}^t
  \! + \! \sum\limits_{a=1}^A \frac{P_{r,a^{*}}^a}{P_{r,b^{*}}^{a^{*}}} \varsigma^2
  \bar{\bm{R}}_{\text{r},r}^t\right)^{-1} 
  \varsigma^2 \bar{\bm{R}}_{\text{r},r}^t ,
\end{align}
 { in which $\varsigma^2\bar{\bm{R}}_{\text{r},r}^t$ is the channel's covariance matrix and
 the channel's spatial correlation matrix} $\bar{\bm{R}}_{\text{r},r}^t$ is given by
\begin{align}\label{eq22}
 \bar{\bm{R}}_{\text{r},r}^t =& \bar{\bm{R}}_t \otimes
  \bar{\bm{R}}^r.
\end{align}
 { Finally, the covariance matrix of the channel estimation error obeys 
 $\bm{vec}\left(\widetilde{\bm{H}}_{b^*}^{a^*}\right)=\bm{vec}\left(\bm{H}_{b^*}^{a^*}\right)-
 \bm{vec}\left(\widehat{\bm{H}}_{b^*}^{a^*}\right)$ by $\bm{\Xi}_{b^*}^{a^*}$. Clearly, we have}
\begin{align}\label{eq23}
 \bm{\Xi}_{b^*}^{a^*} =& \varsigma^2\bar{\bm{R}}_{\text{r},r}^t -
 \bm{\Phi}_{b^{*}}^{a^{*}} \in \mathbb{C}^{N_tN_r \times N_tN_r} ,
\end{align}
 and $\bm{\Xi}_{b^*}^{a^*}$ can be explicitly expressed in the following form
\begin{align}\label{eq24}
 \bm{\Xi}_{b^*}^{a^*} =& \left[\begin{array}{cccc} \left[\bm{\Xi}_{b^*}^{a^*}\right]_{(1,1)} &
  \left[\bm{\Xi}_{b^*}^{a^*}\right]_{(1,2)} & \cdots & \left[\bm{\Xi}_{b^*}^{a^*}\right]_{(1,N_r)} \\
  \vdots & \vdots & \cdots & \vdots \\
  \left[\bm{\Xi}_{b^*}^{a^*}\right]_{(N_r,1)} & \left[\bm{\Xi}_{b^*}^{a^*}\right]_{(N_r,2)} & \cdots
  & \left[\bm{\Xi}_{b^*}^{a^*}\right]_{(N_r,N_r)}
\end{array}\right] \! ,
\end{align}
 where { $\left[\bm{\Xi}_{b^*}^{a^*}\right]_{(i,j)}=\mathcal{E}\left\{
\left[\widetilde{\bm{H}}_{b^*}^{a^*}\right]_{[i:~]}^{\rm H}
 \left[\widetilde{\bm{H}}_{b^*}^{a^*}\right]_{[j:~]}\right\}\in
 \mathbb{C}^{N_t\times N_t}$,} $\forall i,j\in \{1,2,\cdots, N_r\}$. This indicates
 that the distribution of $\left[\widetilde{\bm{H}}_{b^*}^{a^*}\right]_{[n_r:~]}$
 is given by
\begin{align}\label{eq25}
 \left[\widetilde{\bm{H}}_{b^*}^{a^*}\right]_{[n_r:~]}^{\rm T} \sim &
 \mathcal{CN}\left(\bm{0}_{N_t\times 1},\left[\bm{\Xi}_{b^*}^{a^*}\right]_{(n_r,n_r)}\right) ,
\end{align}
 where $\bm{0}_{N_t\times 1}$ is the $N_t$-dimensional zero vector. 

 { Bearing in mind the distribution (\ref{eq20}), the correlation matrix obeys $\mathcal{E}\left\{\bm{vec}\left(
 \widehat{\bm{H}}_{a^{*}}^{b^{*}}\right) \right.$ \\ $\left.\bm{vec}\left(\widehat{\bm{H}}_{a^{*}}^{b^{*}}\right)^{\rm H}\right\}
 =\nu^2\bm{M}_{b^{*}}^{a^{*}}+\bm{\Phi}_{b^{*}}^{a^{*}}$,} where we have
 \begin{align}\label{eq26}
 \bm{M}_{b^{*}}^{a^{*}} =& \bm{vec}\left(\bm{H}_{\text{d},a^{*}}^{b^{*}}\right)
  \bm{vec}\left(\bm{H}_{\text{d},a^{*}}^{b^{*}}\right)^{\rm H} \in \mathbb{C}^{N_tN_r\times N_tN_r} .
\end{align}
 Furthermore, $\bm{M}_{b^{*}}^{a^{*}}$ can be expressed in a form similar to (\ref{eq24}) having
 the $(i,j)$-th sub-matrix { $\left[\bm{M}_{b^{*}}^{a^{*}}\right]_{(i,j)}=\left[
 \bm{H}_{\text{d},a^{*}}^{b^{*}}\right]_{[i:~]}^{\rm H}\left[\bm{H}_{\text{d},a^{*}}^{b^{*}}\right]_{[j:~]}
 \in\mathbb{C}^{N_t\times N_t}$,} $\forall i,j\in \{1,2,\cdots, N_r\}$.

 Hence, upon recalling (\ref{eq11}) and (\ref{eq19}), we have
\begin{align}\label{eq27}
 & \mathcal{E}\! \left\{\! \left[\bm{H}_{b^*}^{a^*}\right]_{[n_r^{*}:~]}
  \left[\bm{V}_{b^*}^{a^*}\right]_{[~:n_r^{*}]}\! \right\} =  \mathcal{E}\! \left\{\! 
  \left[\bm{H}_{b^*}^{a^*}\right]_{[n_r^{*}:~]}
  \left[\widehat{\bm{H}}_{b^*}^{a^*}\right]_{[n_r^{*}:~]}^{\rm H}\! \right\}  \nonumber \\
 & \hspace*{6mm} = \mathcal{E}\left\{\left(\left[\widehat{\bm{H}}_{b^*}^{a^*}\right]_{[n_r^{*}:~]} +
  \left[\widetilde{\bm{H}}_{b^*}^{a^*}\right]_{[n_r^{*}:~]}\right)
  \left[\widehat{\bm{H}}_{b^*}^{a^*}\right]_{[n_r^{*}:~]}^{\rm H} \right\} \nonumber \\ & \hspace*{6mm}
  = \mathcal{E}\left\{ \text{Tr}\left\{ \left[\widehat{\bm{H}}_{b^*}^{a^*}\right]_{[n_r^{*}:~]}^{\rm H}
  \left[\widehat{\bm{H}}_{b^*}^{a^*}\right]_{[n_r^{*}:~]}\right\} \right\} \nonumber \\
 & \hspace*{6mm} = \text{Tr}\left\{\nu^2\left[\bm{M}_{b^{*}}^{a^{*}}\right]_{(n_r^{*},n_r^{*})}
  + \left[\bm{\Phi}_{b^{*}}^{a^{*}}\right]_{(n_r^{*},n_r^{*})}\right\} ,
\end{align}
 where $\text{Tr}\{ ~\}$ denotes the matrix trace operator, and
 $\left[\bm{\Phi}_{b^{*}}^{a^{*}}\right]_{(i,j)}\in\mathbb{C}^{N_t\times N_t}$,
 $\forall i,j\in \{1,2,\cdots, N_r\}$, is the $(i,j)$-th sub-matrix of $\bm{\Phi}_{b^{*}}^{a^{*}}$
 which has a structure similar to (\ref{eq24}). Thus, by denoting
\begin{align}\label{eq28}
 \left[\bm{\Theta}_{b^{*}}^{a^{*}}\right]_{(n_r^{*},n_r^{*})} =& \nu^2
  \left[\bm{M}_{b^{*}}^{a^{*}}\right]_{(n_r^{*},n_r^{*})} + \left[\bm{\Phi}_{b^{*}}^{a^{*}}\right]_{(n_r^{*},n_r^{*})} ,
\end{align}
 we have
\begin{align}\label{eq29}
 \left|\mathcal{E}\left\{\left[\bm{H}_{b^*}^{a^*}\right]_{[n_r^{*}:~]}
  \left[\bm{V}_{b^*}^{a^*}\right]_{[~:n_r^{*}]}\right\}\right|^{2}\! =&
 \left(\text{Tr}\left\{\left[\bm{\Theta}_{b^{*}}^{a^{*}}\right]_{(n_r^{*}, n_r^{*})}
  \right\}\right)^2 \! .
\end{align}
 { Note that multiplying (\ref{eq29}) with $P_{r,b^*}^{a^*}$ leads to the desired signal
 power, i.e. the numerator of the SINR expression (\ref{eq15}).}

 As shown in the Appendix, as $N_t\to \infty$, we have\setcounter{equation}{29}
\begin{align}\label{eq30}
 & \text{Var}\left\{\left[\bm{H}_{b^*}^{a^*}\right]_{[n_r^{*}:~]}
  \left[\bm{V}_{b^*}^{a^*}\right]_{[~:n_r^{*}]}\right\} 
 =\text{Tr}\left\{\left[\bm{\Xi}_{b^{*}}^{a^{*}}\right]_{(n_r^{*}, n_r^{*})}
  \left[\bm{\Theta}_{b^{*}}^{a^{*}}\right]_{(n_r^{*}, n_r^{*})} \right\} .
\end{align}
 Additionally, recalling (\ref{eq9}) and after some simplifications, we can express \\ 
 $\mathcal{E}\Big\{\left|\left[\bm{H}_{b^*}^{a^*}\right]_{[n_r^{*}:~]}
 \left[\bm{V}_{b^*}^{a^*}\right]_{[~:n_r]}\right|^{2}\Big\}$ as
\begin{align}\label{eq31}
 & \mathcal{E}\left\{\left|\left[\bm{H}_{b^*}^{a^*}\right]_{[n_r^{*}:~]}
  \left[\bm{V}_{b^*}^{a^*}\right]_{[~:n_r]}\right|^{2}\right\} 
  = \text{Tr}\Bigg\{\bigg(\nu^2\left[\bm{M}_{b^{*}}^{a^{*}}\right]_{(n_r,n_r)} 
 + \Big(\left[\bm{\Phi}_{b^{*}}^{a^{*}}\right]_{(n_r,n_r)} +
  \frac{\sigma_w^2}{P_{r,b^{*}}^{a^{*}}}\bm{I}_{N_t} \nonumber \\
 & \hspace*{1mm}+
  \sum\limits_{a=1}^A\frac{P_{r,a^{*}}^a}{P_{r,b^{*}}^{a^{*}}}\varsigma^2
  \left[\bar{\bm{R}}_{\text{r},r}^t\right]_{(n_r,n_r)} \Big)  \left[\bm{\Omega}_{b^{*}}^{a^{*}}\right]_{(n_r,n_r)} \Bigg) 
  \bigg(\nu^2\left[\bm{M}_{b^{*}}^{a^{*}}\right]_{(n_r^{*},n_r^{*})} +
  \varsigma^2\left[\bar{\bm{R}}_{\text{r},r}^t\right]_{(n_r^{*},n_r^{*})} \bigg)\Bigg\} ,
\end{align}
 asymptotically, where $\left[\bm{\Omega}_{b^{*}}^{a^{*}}\right]_{(n_r,n_r)}\in
\mathbb{C}^{N_t\times N_t}$ is given by
\begin{align}\label{eq32}
  \left[\bm{\Omega}_{b^{*}}^{a^{*}}\right]_{(n_r,n_r)} =& \varsigma^2\left[\bar{\bm{R}}_{\text{r},r}^t\right]_{(n_r,n_r)}
  \Big(\frac{\sigma_w^2}{P_{r,b^{*}}^{a^{*}}}\bm{I}_{N_t}
  + \left[\bar{\bm{R}}_{\text{r},r}^t\right]_{(n_r,n_r)} 
  + \sum\limits_{a=1}^A\frac{P_{r,a^{*}}^a}{P_{r,b^{*}}^{a^{*}}}\varsigma^2
  \left[\bar{\bm{R}}_{\text{r},r}^t\right]_{(n_r,n_r)} \Big)^{-1} ,
\end{align}
 in which $\left[\bar{\bm{R}}_{\text{r},r}^t\right]_{(i,j)}\in\mathbb{C}^{N_t\times N_t}$,
 $\forall i,j\in \{1,2,\cdots, N_r\}$, represents the $(i,j)$-th sub-matrix of
 $\bar{\bm{R}}_{\text{r},r}^t$, which has a structure similar to (\ref{eq24}).
 Similarly, we can express $\mathcal{E}\left\{\left|\left[\bm{H}_{b^*}^a\right]_{[n_r^{*}:~]}
 \left[\bm{V}_{b^a}^a\right]_{[~:n_r]}\right|^{2}\right\}$ asymptotically as
\begin{align}\label{eq33}
 & \mathcal{E}\left\{\left|\left[\bm{H}_{b^*}^a\right]_{[n_r^{*}:~]}\left[\bm{V}_{b^a}^a\right]_{[~:n_r]}\right|^2\right\}
  = \text{Tr}\Bigg\{\Bigg(\nu^2\left[\bm{M}_{b^{a}}^a\right]_{(n_r,n_r)} 
 + \! \Big(\left[\bm{\Phi}_{b^a}^a\right]_{(n_r,n_r)} +
  \sum\limits_{a=1}^A\frac{P_{r,a^*}^a\varsigma^2}{P_{r,b^*}^{a^*}}
  \left[\bar{\bm{R}}_{\text{r},r}^t\right]_{(n_r,n_r)}\!\! \nonumber \\ & \hspace*{1mm}+\! \frac{\sigma_w^2}{P_{r,b^{*}}^{a^{*}}}
  \bm{I}_{N_t}\! \Big)   
 \left[\bm{\Omega}_{b^a}^a\right]_{(n_r,n_r)} \! \Bigg) \left(\nu^2
  \left[\bm{M}_{b^{*}}^a\right]_{(n_r^{*},n_r^{*})} \! + \! \varsigma^2
  \left[\bar{\bm{R}}_{\text{r},r}^t\right]_{(n_r^{*},n_r^{*})} \right)\! \Bigg\} .
\end{align}

 Upon substituting (\ref{eq30}), (\ref{eq31}) and (\ref{eq33}) into (\ref{eq16}),
 therefore, we arrive asymptotically at the power of the interference plus noise
 $P_{{\text{I\&N}}_{b^*,n_r}^{a^*}}$ given by
\begin{align}\label{eq34}
 & P_{{\text{I\&N}}_{b^*,n_r}^{a^*}} = P_{r,b^{*}}^{a^{*}}
  \text{Tr}\left\{\left[\bm{\Xi}_{b^{*}}^{a^{*}}\right]_{(n_r^{*},n_r^{*})}
  \left[\bm{\Theta}_{b^{*}}^{a^{*}}\right]_{(n_r^{*},n_r^{*})}\right\} + P_{r,b^{*}}^{a^{*}}
  \sum\limits_{\begin{subarray}{ll}n_r = 1 \\ n_r \neq n_r^{*}\end{subarray}}^{N_r}
  \text{Tr}\Bigg\{\bigg(\nu^2\left[\bm{M}_{b^{*}}^{a^{*}}\right]_{(n_r,n_r)} +\! \Big(\!\! \left[\bm{\Phi}_{b^{*}}^{a^{*}}\right]_{(n_r,n_r)}\! \nonumber \\
 & \hspace*{2mm}
  +\!
 \frac{\sigma_w^2}{P_{r,b^{*}}^{a^{*}}}\bm{I}_{N_t}+\! \frac{A \bar{P}_r}{P_{r,b^{*}}^{a^{*}}}\varsigma^2\left[\bar{\bm{R}}_{\text{r},r}^t\right]_{(n_r, n_r)}
  \! \Big)\!\! \left[\bm{\Omega}_{b^{*}}^{a^{*}}\right]_{(n_r,n_r)}\! \bigg)\! \Big(\! \nu^2\! \left[\bm{M}_{b^{*}}^{a^{*}}
  \right]_{(n_r^{*},n_r^{*})}\! +\! \varsigma^2\left[\bar{\bm{R}}_{\text{r},r}^t\right]_{(n_r^{*},n_r^{*})}
  \! \Big) \!\! \Bigg\} \! \nonumber \\
 & \hspace*{2mm}   
  +\! \bar{P}_r\sum\limits_{a=1}^A\sum\limits_{n_r=1}^{N_r}\text{Tr}\Bigg\{\!\!\bigg(\! \nu^2
  \left[\bm{M}_{b^a}^a\right]_{(n_r,n_r)} \!+\! \Big(\left[\bm{\Phi}_{b^a}^a\right]_{(n_r,n_r)}\! +\! \frac{A \bar{P}_r}{P_{r,b^{*}}^{a^{*}}}
  \varsigma^2\left[\bar{\bm{R}}_{\text{r},r}^t\right]_{(n_r,n_r)}\! +\! \frac{\sigma_w^2}{P_{r,b^{*}}^{a^{*}}}
  \bm{I}_{N_t}\Big)\left[\bm{\Omega}_{b^a}^a\right]_{(n_r,n_r)}\!\! \bigg) \nonumber \\ & \hspace*{2mm}
  \times \Big(\nu^2
  \left[\bm{M}_{b^{*}}^a\right]_{(n_r^{*},n_r^{*})} + \varsigma^2\left[\bar{\bm{R}}_{\text{r},r}^t
  \right]_{(n_r^{*},n_r^{*})} \Big)\Bigg\} + \sigma_w^2.
 \end{align}
 Furthermore, substituting (\ref{eq29}) into  (\ref{eq15}) leads to the following
 asymptotic SINR expression \setcounter{equation}{34}
\begin{align}\label{eq35}
 \gamma_{b^{*},n_{r}}^{a^{*}} =& \frac{P_{r,b^{*}}^{a^{*}}\left(
  \text{Tr}\left\{\left[\bm{\Theta}_{b^{*}}^{a^{*}}\right]_{(n_r^{*},n_r^{*})}\right\}
  \right)^2}{P_{{\text{I\&N}}_{b^*,n_r}^{a^*}}} .
\end{align}

\begin{remark}\label{RM1}
 Both $\left[\bm{M}_{b^a}^a\right]_{(n_r,n_r)}$ and $\left[\bm{M}_{b^{*}}^a\right]_{(n_r^{*},n_r^{*})}$
 in (\ref{eq34}) are unavailable to aircraft $a^*$, since there is no cooperation between
 the related aircraft. However, both these two terms can be substituted by
 $\left[\bm{M}_{b^{*}}^{a^{*}}\right]_{(n_r^{*}, n_r^{*})}$ as a reasonable approximation.
 The simulation results presented in Section~\ref{S4} will demonstrate that this approximation
 is sufficiently accurate.
\end{remark}

{
\begin{remark}\label{RM2}
 The high velocity of aircraft results in rapidly fluctuating fading. The channel
 estimator (\ref{eq9}) is efficient, since mean square
 error matches the Cram\'{e}r-Rao lower bound. However, 
 by the time the transmitter transmits the precoded signal based on this channel estimate, the
 real channel may have changed. This mobility-induced channel estimation `error' causes a performance degradation
 w.r.t. the optimal performance based on perfect channel estimate. An effective approach to
 mitigate this performance degradation owing to channel estimation errors is to adopt robust transmit 
 precoding. The design of robust precoding is
 beyond the scope of this paper. Some highly effective robust precoding designs can be
 found in \cite{Gong_etal2017,Gong_etal2017a}.
\end{remark}
}

{
\begin{remark}\label{RM3}
As our derivation does not impose any specific geometric structure on the antenna array,
 our results and therefore our proposed physical-layer transmission scheme is applicable to
  systems equipped with uniformly spaced linear arrays (ULAs), uniformly spaced rectangular
 arrays (URAs), or any other generic antenna arrays.
\end{remark}
}

\subsection{Distance-based ACM scheme}\label{S3.2}

 The proposed distance-based ACM scheme is illustrated in Fig.~\ref{FIG2}. Our scheme
 explicitly uses the distance $d_{b^*}^{a^*}$ between aircraft $a^*$ and aircraft $b^{*}$
 as the switching metric to adapt the modulation mode and code rate. Similar to the
 conventional ACM scheme, our distance-based ACM also consists of the set of $K$ ACM modes,
 but its switching thresholds comprise the $K$ distance thresholds $\{d_k\}_{k=1}^K$. An
 example of this distance-based ACM scheme using $K=7$ is given in Table~\ref{Tab1}, where
 the SE is calculated as: $\log_2(\text{modulation order}) \times \text{code rate}\times
 [{(N-N_{\rm cp})}/{N}]$, and the data rate per DRA is calculated as:
 $\text{spectral efficiency}\times B_{\text{total}}$, in which $B_{\text{total}}$ is the
 bandwidth available, while the total data rate is given by: $\text{data rate per DRA}
 \times N_r$. The modulation schemes and code rates are adopted from the design of
 VersaFEC \cite{VersaFEC}, which covers a family of 12 short-block LDPC code rates
 with the matched modulation schemes. VersaFEC is specifically designed for low latency
 and ACM applications. The operations of this distance-based ACM are now given below.

\begin{figure}[tbp!]
\vspace*{-1mm}
\begin{center}
 \includegraphics[width=0.9\columnwidth]{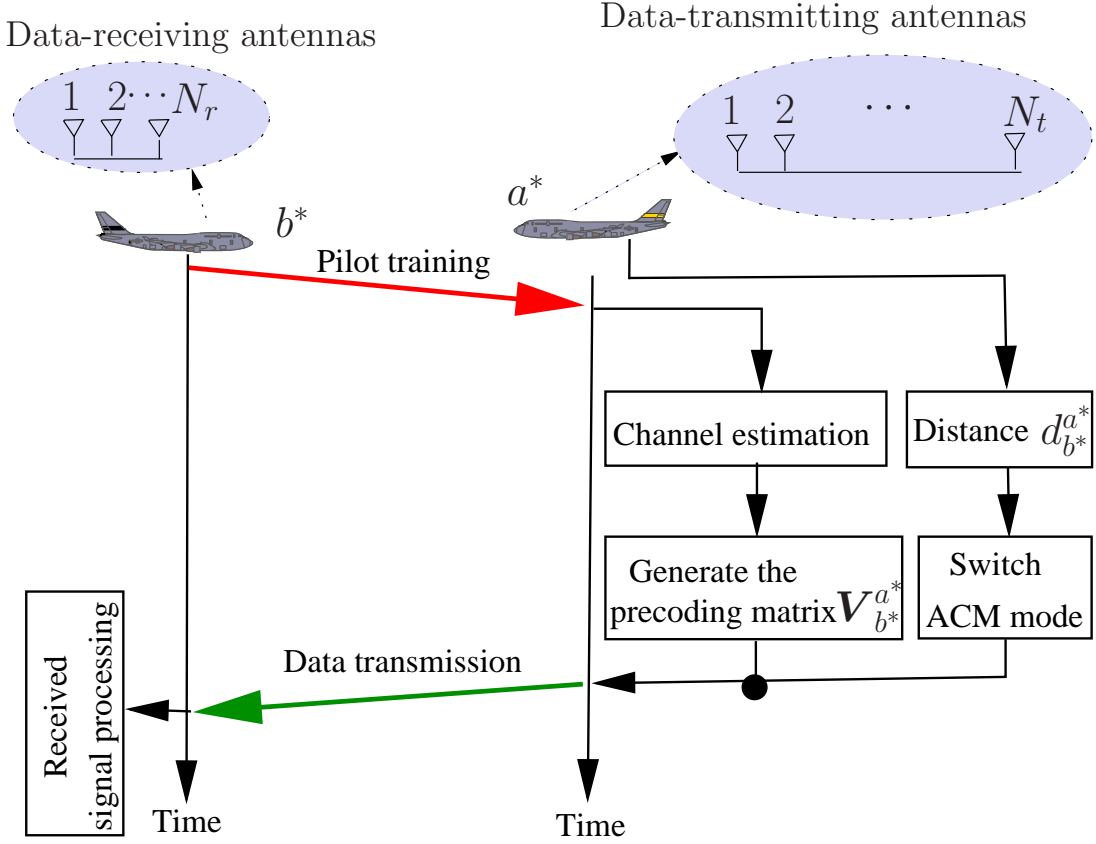}
\end{center}
\vspace*{-5mm}
\caption{Schematic of the proposed distance-based adaptive coding and modulation scheme.}
\label{FIG2}
\vspace*{-4mm}
\end{figure}

\vspace*{1mm}

\noindent
 1)~Aircraft $a^{*}$ estimates the channel $\bm{H}_{b^{*}}^{a^{*}}$ based on the pilots
 transmitted by aircraft $b^{*}$, as detailed in Section~\ref{S2.1}.

\vspace*{1mm}

\noindent
 2)~Aircraft $a^{*}$ calculates the TPC matrix $\bm{V}_{b^{*}}^{a^{*}}$ according to
 (\ref{eq11}).

\vspace*{1mm}

\noindent
 3)~Aircraft $a^{*}$ selects an ACM mode to transmit the data according to
\begin{align}\label{eq36}
 & \text{ If } d_{k} \le d_{b^{*}}^{a^{*}} < d_{k-1}: \text{ choose mode } k , 
\end{align}
 where $k\in \{1,2,\cdots , K\}$, and we assume $d_0=D_{\max}$. When
 $d_{b^{*}}^{a^{*}}\ge D_{\max}$, there exists no available communication link,
 since the two aircraft are beyond each others' communication range. Since the
 minimum flight-safety separation must be obeyed, we do not consider the senario
 of $d_{b^*}^{a^*} \le D_{\min}$.

\begin{table*}[bp!]
\footnotesize
\vspace*{-3mm}
\caption{An example of adaptive coding and modulation scheme with $N_t=32$ and $N_r=4$. The
 other system parameters for this ACM are listed in Table~\ref{Tab3}.}
\vspace*{-4mm}
\begin{center}
\begin{tabular}{|C{1.2cm}|C{1.4cm}|C{1.4cm}|C{2.4cm}|C{2.4cm}|C{2.4cm}|C{2.3cm}|}
\hline
 Mode $k$ & Modulation & Code rate & Spectral efficiency (bps/Hz) & Switching threshold $d_k$ (km)
 & Data rate per receive antenna (Mbps) & Total data rate (Mbps) \\ \hline
 1 & BPSK   & 0.488 & 0.459 & 500 & 2.754  & 11.016 \\ \hline
 2 & QPSK   & 0.533 & 1.000 & 350 & 6.000  & 24.000 \\ \hline
 3 & QPSK  & 0.706 & 1.322 & 200 & 7.932 & 31.728 \\ \hline
 4 & 8-QAM  & 0.642 & 1.809 & 110  & 10.854 & 43.416 \\ \hline
 5 & 8-QAM & 0.780 & 2.194 & 40  & 13.164 & 52.656 \\ \hline
 6 & 16-QAM & 0.731 & 2.747 & 25  & 16.482 & 65.928 \\ \hline
 7 & 16-QAM & 0.853 & 3.197 & 5.56  & 19.182 & 76.728 \\ \hline
\end{tabular}
\end{center}
\label{Tab1}
\vspace*{-2mm}
\end{table*}

\begin{table*}[bp!]
{
\footnotesize
\vspace*{-2mm}
\caption{{ An example of adaptive coding and modulation scheme with $N_t=64$ and $N_r=4$. The
 other system parameters for this ACM are listed in Table~\ref{Tab3}.}}
\vspace*{-4mm}
\begin{center}
\begin{tabular}{|C{1.2cm}|C{1.4cm}|C{1.4cm}|C{2.4cm}|C{2.4cm}|C{2.4cm}|C{2.3cm}|}
\hline
 Mode $k$ & Modulation & Code rate & Spectral efficiency (bps/Hz) & Switching threshold $d_k$ (km)
 & Data rate per receive antenna (Mbps) & Total data rate (Mbps) \\ \hline
 1 & QPSK  & 0.706 & 1.322 & 400 & 7.932 & 31.728 \\ \hline
 2 & 8-QAM  & 0.642 & 1.809 & 250  & 10.854 & 43.416 \\ \hline
 3 & 8-QAM & 0.780 & 2.194 & 120  & 13.164 & 52.656 \\ \hline
 4 & 16-QAM & 0.731 & 2.747 & 50  & 16.482 & 65.928 \\ \hline
 5 & 16-QAM & 0.853 & 3.197 & 5.56  & 19.182 & 76.728 \\ \hline
\end{tabular}
\end{center}
\label{Tab2}
}
\vspace*{-2mm}
\end{table*}

 The switching distance threshold for each ACM mode is determined based on the
 achievable rate per DRA as a function of distance. Specifically, the theoretically
 achievable rate per DRA as a function of distance is calculated using (\ref{eq15}).
 The distance thresholds $\{d_k\}_{k=1}^K$ are chosen so that the SE of mode $k$
 is lower than the theoretically achievable rate per DRA in the distance range of
 $[d_k , ~ d_{k-1}]$ to ensure successful transmission. Fig.~\ref{FIG3a} illustrates
 how the 7 distance thresholds are designed for the example provided in Table~\ref{Tab1}.
 For this example, $N_t=32$, $N_r=4$, and the total system bandwidth is
 $B_{\rm total}=6$\,MHz which is reused by every aircraft in the system. The theoretically
 achievable rate per DRA as a function of distance is depicted as the dot-marked solid
 curve in Fig.~\ref{FIG3a}. By designing the 7 distance thresholds $d_k$ for
 $1\le k\le 7$ to ensure that the SE of mode $k$ is lower than the theoretically
 achievable rate in the distance range $[d_k , ~ d_{k-1}]$, we obtain the 7 desired
 distance thresholds for this ACM example, which are indicated in Fig.~\ref{FIG3a} as
 well as listed in Table~\ref{Tab1}. The designed 7 ACM modes are: \emph{mode 1}
 having the SE of 0.459 bps/Hz, \emph{mode 2} with the SE 1.000 bps/Hz, \emph{mode 3}
 with the SE 1.322 bps/Hz, \emph{mode 4} with the SE 1.809 bps/Hz, \emph{mode 5} with
 the  SE 2.194 bps/Hz, \emph{mode 6} with the 2.747 bps/Hz, and \emph{mode 7} with the
 SE 3.197 bps/Hz.

\begin{figure*}[htp!]
\vspace*{-2mm}
\begin{center}
\subfigure[$N_t = 32, N_r = 4$]{\includegraphics[width=0.47\textwidth]{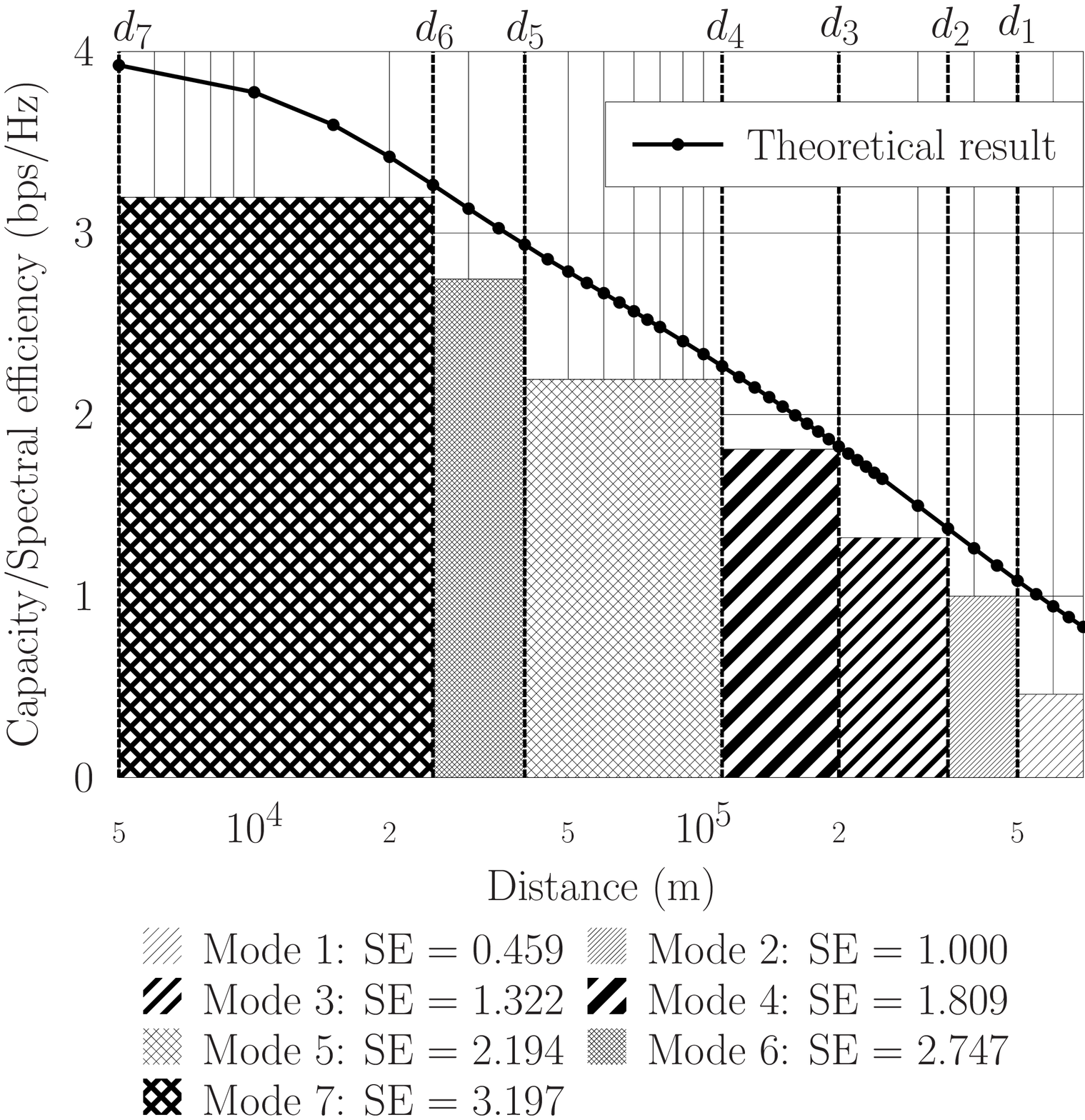} 
\label{FIG3a}}
\subfigure[{ $N_t = 64, N_r = 4$}]{\includegraphics[width=0.47\textwidth]{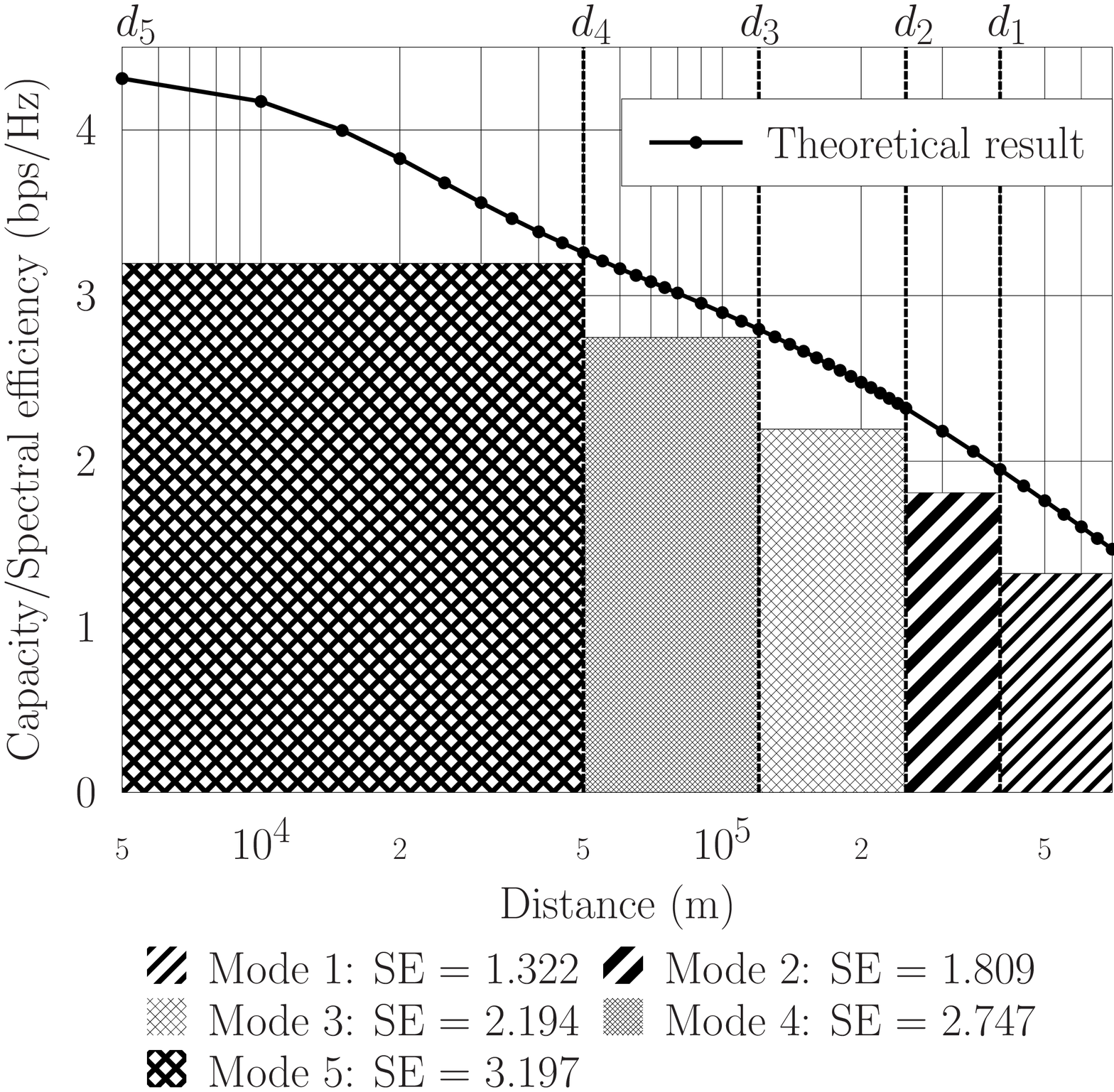}
\label{FIG3b}}
\end{center}
\vspace{-4mm}
\caption{Illustration of how to obtain the desired distance thresholds for the
 proposed distance-based ACM scheme, using { the examples of Tables~\ref{Tab1} and
\ref{Tab2}.}}
\label{FIG3}
\vspace{-4mm}
\end{figure*}

{
 We also provide another design example of the ACM scheme for $K=5$ modes. For this
 example, we have $N_t=64$ and $N_r=4$, while the other system parameters are the same as the
 ACM scheme listed in Table~\ref{Tab1}. As shown in Fig.~\ref{FIG3b} and Table~\ref{Tab2},
 the five ACM modes have the SEs of 1.322\,bps/Hz, 1.809\,bps/Hz, 2.194\,bps/Hz,
 2.747\,bps/Hz and 3.197\,bps/Hz, respectively, while the corresponding switching
 thresholds are 400\,km, 250\,km, 120\,km, 50\,km, and 5.56\,km, respectively. By comparing
 Table~\ref{Tab2} to Table~\ref{Tab1}, it can be seen that employing a larger number
 of transmit antennas enables an ACM mode to operate over a larger range of distances, or
 it allows the system to transmit at a higher SE over a given communication distance. This
 makes sense because a well-known MIMO property is that employing more transmit antennas
 can mitigate the interference more effectively.
}

\begin{remark}\label{RM4}
 It is worth recapping that the conventional instantaneous SNR-based ACM is unsuitable
 for aeronautical communication applications, because the speed of aircraft is ultra high,
 which results in rapidly changing of large-scale fading and consequently very unreliable
 estimate of the instantaneous SNR as well as leads to frequently switching among the
 modes. Using erroneous instantaneous SNR estimates to frequently switch modes will
 cause frequent unsuccessful transmissions. By contrast, the proposed ACM scheme switches
 its mode based on the distance, which is readily available to the transmitting aircraft,
 since every jumbo jet has a radar and is equipped with GPS. It can be seen that this
 distance-based ACM scheme is particularly suitable for aeronautical communication
 applications. { Furthermore, it is worth emphasizing that using the distance as the
 switching metric is theoretically well justified, because for the aeronautical communication
 channel the achievable capacity is mainly dependent on the distance, as we have
 analytically derived in Subsection~\ref{S3.1}.} Note that owing to the high velocity of
 the aircraft, no physical layer transmission scheme can guarantee successful transmission
 for every transmission slot. Other higher-layer measures, such as Automatic-Repeat-reQuest
 (ARQ) \cite{chen2014asurvey,ngo2014hybrid}, can be employed for enhancing reliable
 communication among aircraft. Discussing these higher-layer techniques is beyond the scope
 of this paper.
\end{remark}

\begin{table}[tp!]
\caption{Default parameters used in the simulated aeronautical communication system.}
\vspace*{-5mm}
\begin{center}
\begin{tabular}{l|l}
\hline\hline
\multicolumn{2}{l}{System parameters for ACM} \\ \hline
 Number of interference aircraft $A$        & 4 \\ \hline
 Number of DRAs $N_r$    & 4 \\ \hline
 Number of DTAs $N_t$ & 32 \\ \hline
 Transmit power per antennas $P_t$          & 1 watt \\ \hline
 Number of total subcarriers $N$            & 512 \\ \hline
 Number of CPs $N_{\rm cp}$       & 32 \\ \hline
 Rician factor $K_{\text{Rice}}$            & 5  \\ \hline
 Bandwidth $B_{\text{total}}$               & 6 MHz  \\ \hline
 Frequency of centre subcarrier             & 5 GHz \\ \hline\hline
\multicolumn{2}{l}{Other system parameters} \\ \hline
 Correlation factor between antennas $\rho$ & 0.1\\ \hline
 Noise figure at receiver $F$               & 4 dB \\ \hline
 Distance between communicating aircraft $a^{*}$ and $b^{*}$ $d_{b^{*}}^{a^{*}}$  & 10 km \\ \hline
 Maximum communication distance $D_{\max}$  & 740 km \\ \hline\hline
\end{tabular}
\end{center}
\label{Tab3}
\vspace*{-5mm}
\end{table}

\section{Simulation study}\label{S4}

 The default values of the parameters used for our simulated aeronautical communication
 system are summarized in Table~\ref{Tab3}. Unless otherwise specified, these default
 values are used. The number of DRAs is much lower than the number of DTAs in order to
 ensure that the DRAs' signals remain uncorrelated to avoid the interference
 among the antennas at the receiver. The deterministic part of the Rician channel, which
 satisfies $\text{Tr}\left\{\bm{H}_{{\rm d},b}^a\bm{H}_{{\rm d},b}^{a, \rm H}\right\}
 =N_t N_r$, is generated in every Monte-Carlo simulation. The scattering component of
 the Rician channel $\bm{H}_{\text{r}, b}^{a}\in \mathbb{C}^{N_{r} \times N_{t}}$ is
 generated according to
\begin{align}\label{eq37}
 \bm{H}_{\text{r}, b}^{a} =& \bm{R}_{b}\bm{G}\bm{R}^{a} ,
\end{align}
 where $\bm{R}_b=\bm{I}_{N_r}$, since the DRAs are uncorrelated, but the
 $m$th-row and $n$th-column element of the correlation matrix of $\bm{R}^a$, denoted by
 $ [\bm{R}^{a}]_{[m,n]}$, is generated according to{ \cite{martin2004asymptotic,lee2015antenna}}
\begin{align}\label{eq38}
 [\bm{R}^{a}]_{[m,n]} =& \big([\bm{R}^{a}]_{[n,m[}\big)^{\ddagger} = \left(t \rho\right)^{|m - n|} ,
\end{align}
 in which $( ~ )^{\ddagger}$ denotes conjugate operation, and $t\sim {\cal CN}(0,1)$ 
{ is the magnitude of the correlation coefficient, which is determined by the antenna element
 spacing \cite{lee2015antenna}. The antenna array correlation matrix model (\ref{eq38})
 is derived based on the ULA, and this implies that we adopt the ULA in our simulation
 study. However, it is worth recalling Remark~\ref{RM3} stating that our scheme is not restricted
 to the ULA.} In the investigation of the achievable throughput, `Theoretical' is the
 throughput calculated using (\ref{eq17}) relying on the perfect knowledge of
 $\left[\bm{M}_{b^{a}}^{a}\right]_{(n_r, n_r)}$ and $\left[\bm{M}_{b^{*}}^{a}\right]_{(n_r^{*}, n_r^{*})}$
 in (\ref{eq34}), and `Approximate' is the throughput calculated using (\ref{eq15})
 with both $\left[\bm{M}_{b^{a}}^{a}\right]_{(n_r, n_r)}$ and $\left[\bm{M}_{b^{*}}^{a}\right]_{(n_r^{*}, n_r^{*})}$
 substituted by $\left[\bm{M}_{b^{*}}^{a^{*}}\right]_{(n_r^{*}, n_r^{*})}$ in (\ref{eq34}),
 while `Simulation' is the Monte-Carlo simulation result.

\begin{figure*}[htbp!]
\vspace*{-4mm}
\begin{center}
\subfigure[Throughput]{
  \includegraphics[width=0.47\textwidth]{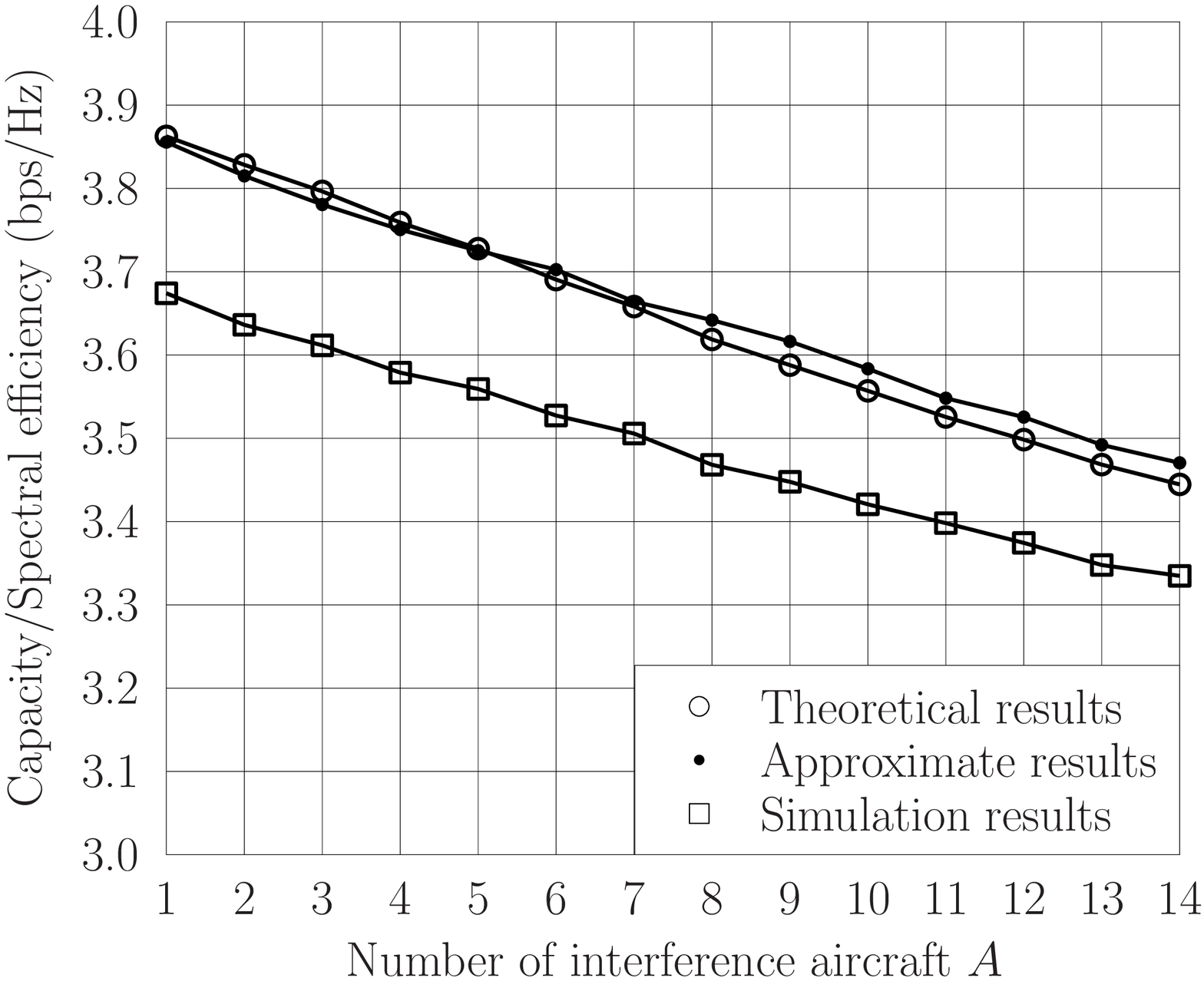} 
  \label{FIG4a}
	}
\subfigure[CCDF]{\includegraphics[width=0.47\textwidth]{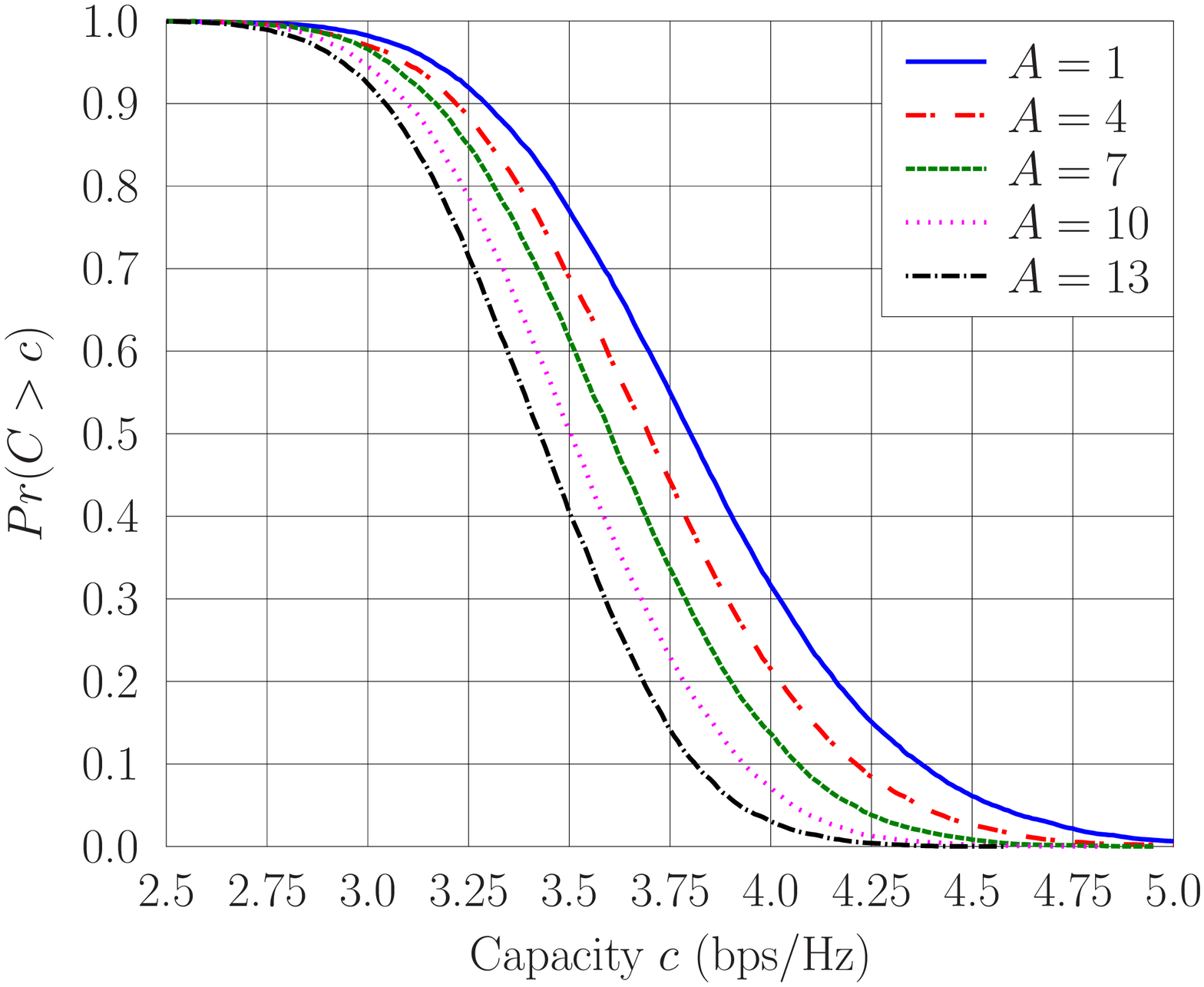}
  \label{FIG4b}}
\end{center}
\vspace{-4mm}
\caption{(a)~The achievable throughput per DRA as a function of the number of
 interfering aircraft $A$, and (b)~The CCDFs of the simulated throughputs per DRA
 for different numbers of interfering aircraft. The distances between the interfering
 aircraft and the desired receiving aircraft are uniformly distributed within the range
 of $\big[d_{b^{*}}^{a^{*}}, ~ D_{\text{max}}\big]$. The rest of the parameters are
 specified in Table~\ref{Tab3}.}
\label{Fig4}
\vspace*{-2mm}
\end{figure*}

\subsection{Results}\label{S4.1}

 Fig.~\ref{FIG4a} investigates the achievable throughput per DRA as a function of the
 number of interfering aircraft $A$, where the `Approximate results' match closely the
 `Theoretical results', confirming that $\left[\bm{M}_{b^{*}}^{a^{*}}\right]_{(n_r^{*},n_r^{*})}$
 is an accurate approximation of both the unknown $\left[\bm{M}_{b^{a}}^{a}\right]_{(n_r,n_r)}$
 and $\left[\bm{M}_{b^{*}}^{a}\right]_{(n_r^{*},n_r^{*})}$. As expected, the achievable
 throughput degrades as the number of interfering aircraft increases. Also the theoretical
 throughput is about 0.2\,bps/Hz higher than the simulated throughput. This is because
 the theoretical throughput is obtained by { using the asymptotic interference plus noise
 power} as $N_t\to \infty$ and, therefore, it represents the asymptotic upper bound of the
 achievable throughput. The complementary cumulative distribution functions (CCDFs) of
 the simulated throughput recorded for different numbers of interfering aircraft $A$ are
 shown in Fig.~\ref{FIG4b}, which characterizes the probability of the achievable throughput
 above a given value.

\begin{figure*}[tbp!]
\vspace*{-4mm}
\begin{center}
\subfigure[Throughput]{
  \includegraphics[width=0.47\textwidth]{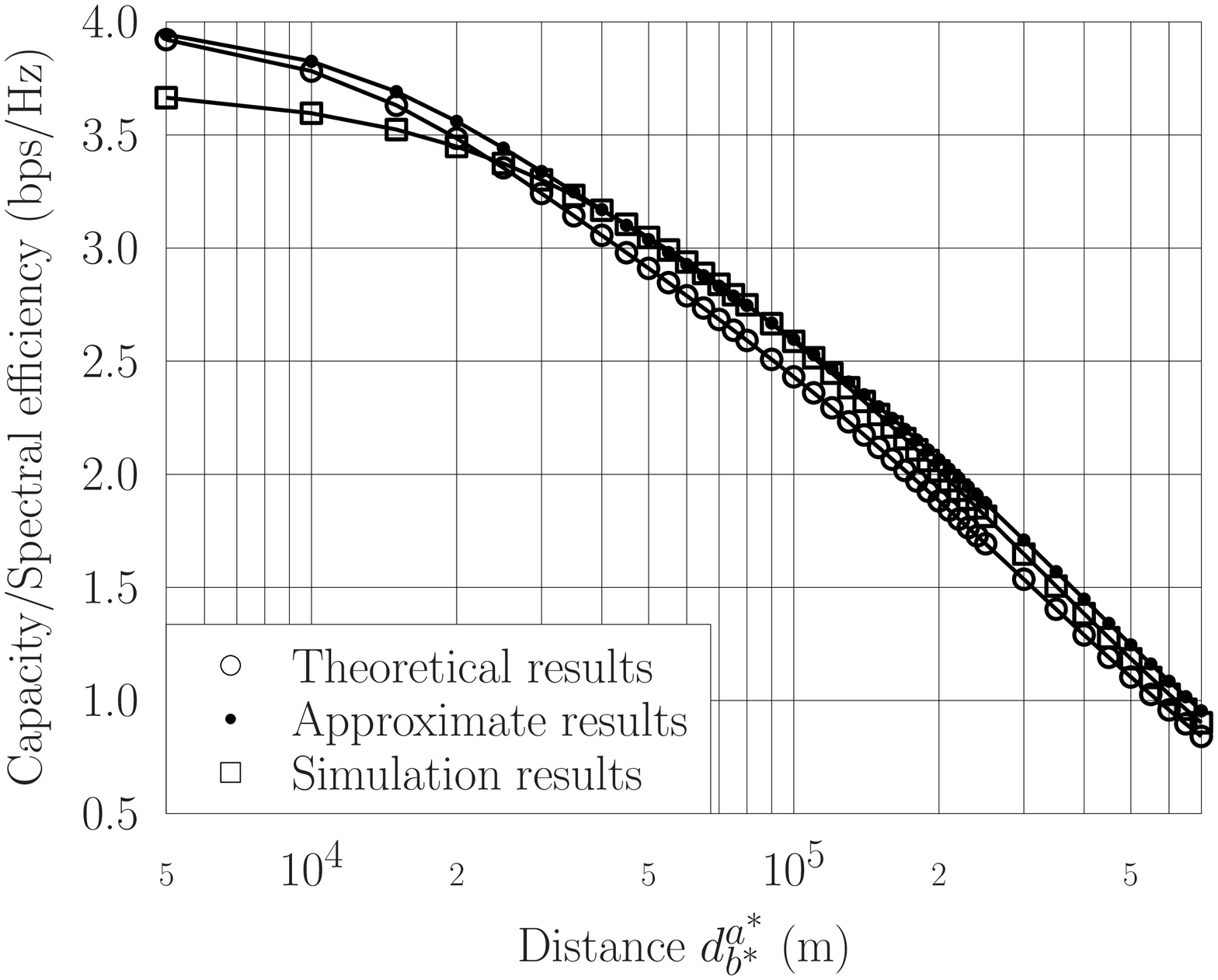} 
  \label{FIG5a}
	}
\subfigure[CCDF]{\includegraphics[width=0.47\textwidth]{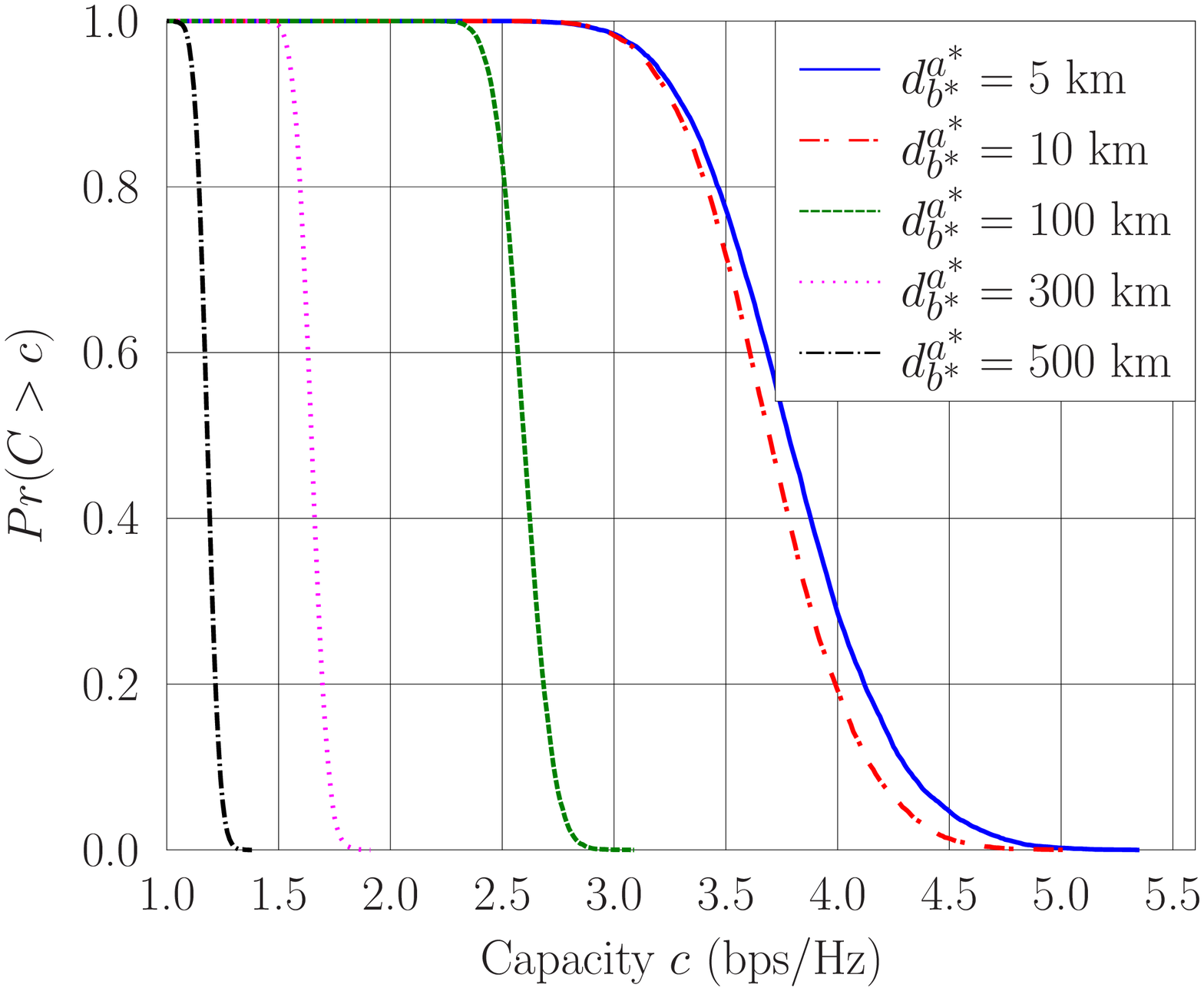}
  \label{FIG5b}}
\end{center}
\vspace{-4mm}
\caption{(a)~The achievable throughput per DRA as a function of the distance
 $d_{b^{*}}^{a^{*}}$ between the desired communicating aircraft $a^{*}$ and $b^{*}$,
 and (b)~The CCDFs of the simulated throughputs per DRA for different
 $d_{b^{*}}^{a^{*}}$. The distances between the interfering aircraft and the
 desired receiving aircraft are uniformly distributed within the range of
 $\big[d_{b^{*}}^{a^{*}},~ D_{\text{max}}\big]$. The rest of the parameters are
 specified in Table~\ref{Tab3}.}
\label{FIG5}
\end{figure*}

 Fig.~\ref{FIG5a} portrays the achievable throughput per DRA as a  function of the distance
 $d_{b^{*}}^{a^{*}}$ between the desired pair of communicating aircraft, while the CCDFs
 of the simulated throughput recorded for different values of $d_{b^{*}}^{a^{*}}$ are
 depicted in Fig.~\ref{FIG5b}. As expected, the achievable throughput degrades upon increasing
 the communication distance. { Observe that the performance gap between the theoretical curve
 and the simulation curve at the point of $d_{b^{*}}^{a^{*}}=10$\,km is also around 0.2\,bps/Hz,
 which agrees with the results of Fig.~\ref{FIG4a}.}

 The impact of the number of DTAs on the achievable throughput is  investigated in
 Fig.~\ref{FIG6}. Specifically, Fig.~\ref{FIG6a} shows the throughput per DRA as a 
 function of the number of DTAs $N_t$, while Fig.~\ref{FIG6b} depicts the CCDFs of the
 simulated throughputs per DRA for different numbers of DTAs $N_t$. Observe that the
 achievable throughput increases as $N_t$ increases. Moreover, when the number of DTAs
 increase to $N_t\ge 120$, the achievable throughput saturates, as seen from
 Fig.~\ref{FIG6a}. Similarly, the CCDFs recorded for $N_t=140$ and $N_t=180$ are
 indistinguishable, as clearly seen from Fig.~\ref{FIG6b}. The  implication is that the
 asymptotic performance is reached for $N_t \ge 120$.

\begin{figure*}[htbp!]
\vspace{-2mm}
\begin{center}
\subfigure[Throughput]{
  \includegraphics[width=0.47\textwidth]{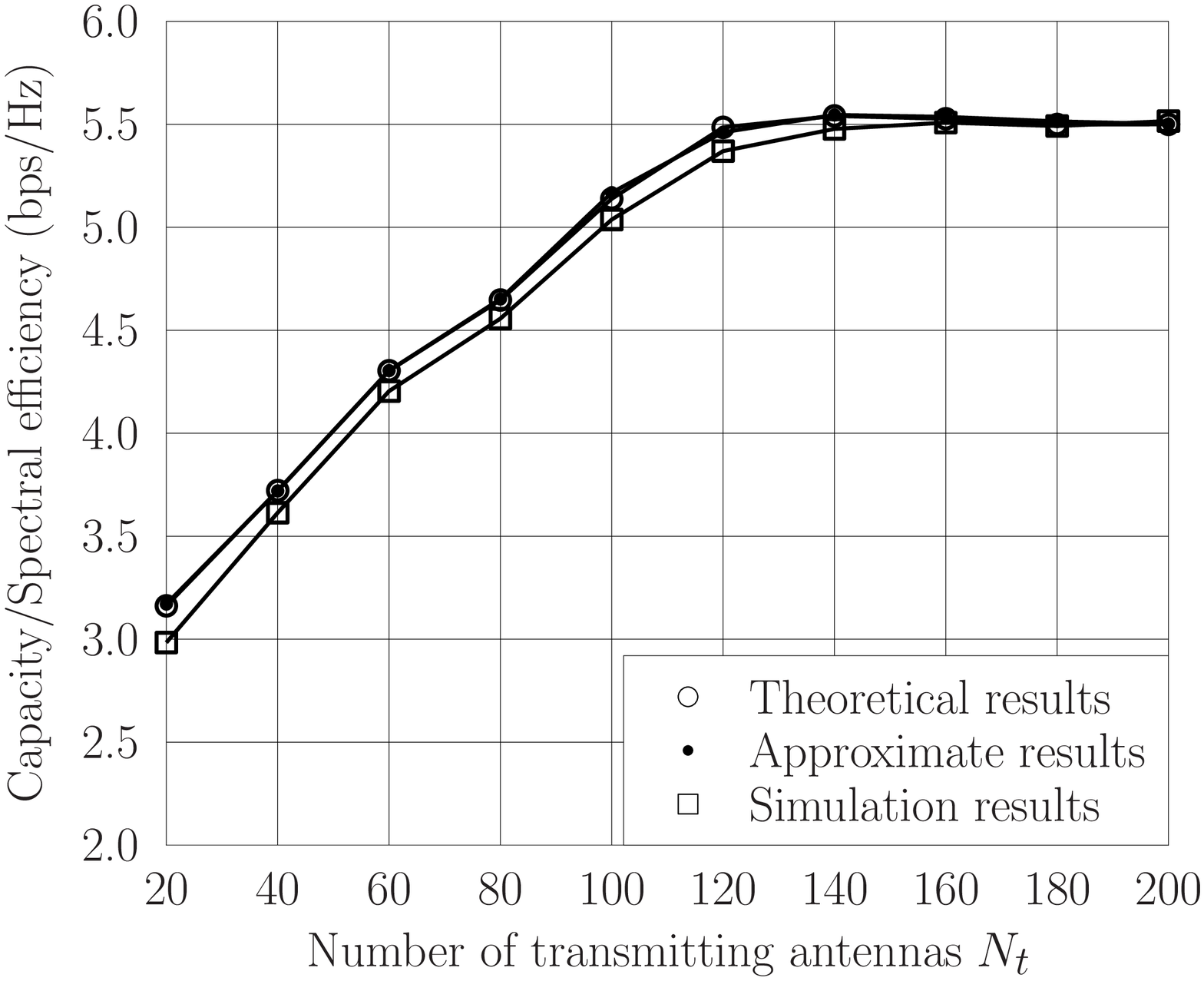} 
  \label{FIG6a}
	}
\subfigure[CCDF]{\includegraphics[width=0.47\textwidth]{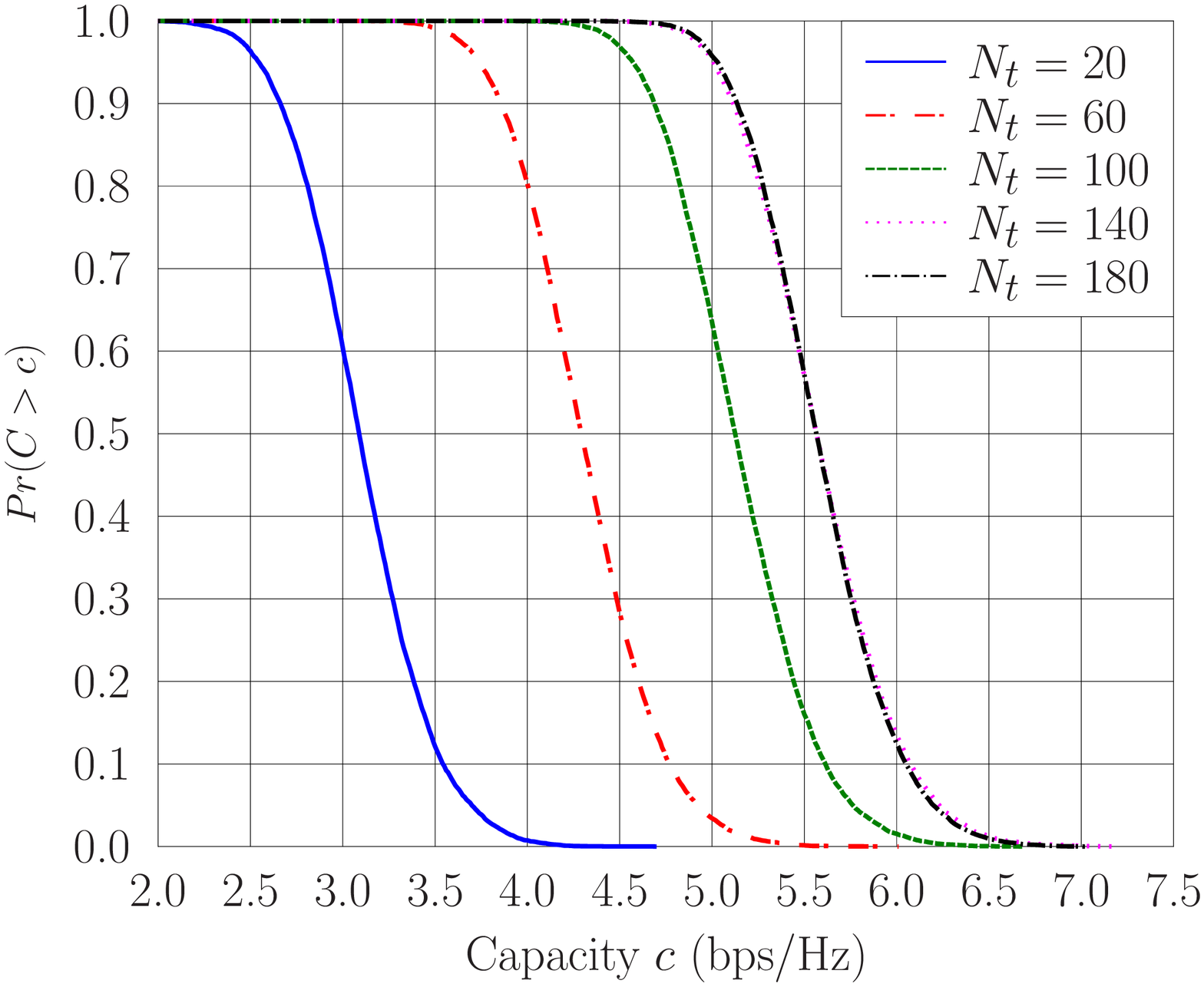}
  \label{FIG6b}}
\end{center}
\vspace{-4mm}
\caption{(a)~The achievable throughput performance per DRA as a function of the
 number of DTAs $N_t$, and (b)~The CCDFs of the simulated throughputs per DRA for
 different numbers of DTAs $N_t$. The distances between the interfering aircraft
 and the desired receiving aircraft are uniformly distributed within the range of
 $\big[d_{b^{*}}^{a^{*}},~ D_{\text{max}}\big]$. The rest of the parameters are
 specified in Table~\ref{Tab3}.}
\label{FIG6}
\vspace{-2mm}
\end{figure*}

\begin{figure*}[htbp!]
\vspace{-4mm}
\begin{center}
\subfigure[Throughput]{
  \includegraphics[width=0.47\textwidth]{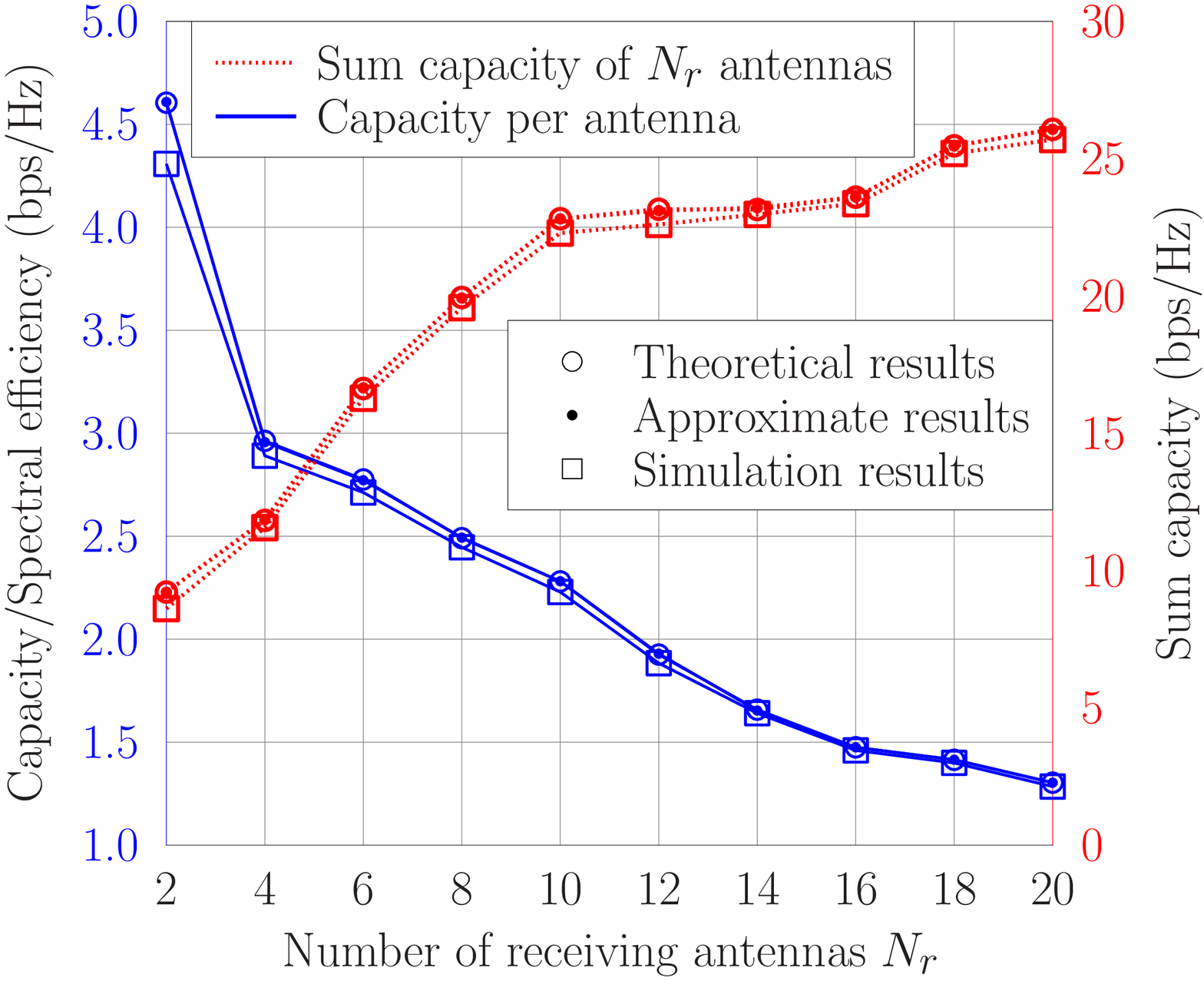} 
  \label{FIG7a}
	}
\subfigure[CCDF]{\includegraphics[width=0.47\textwidth]{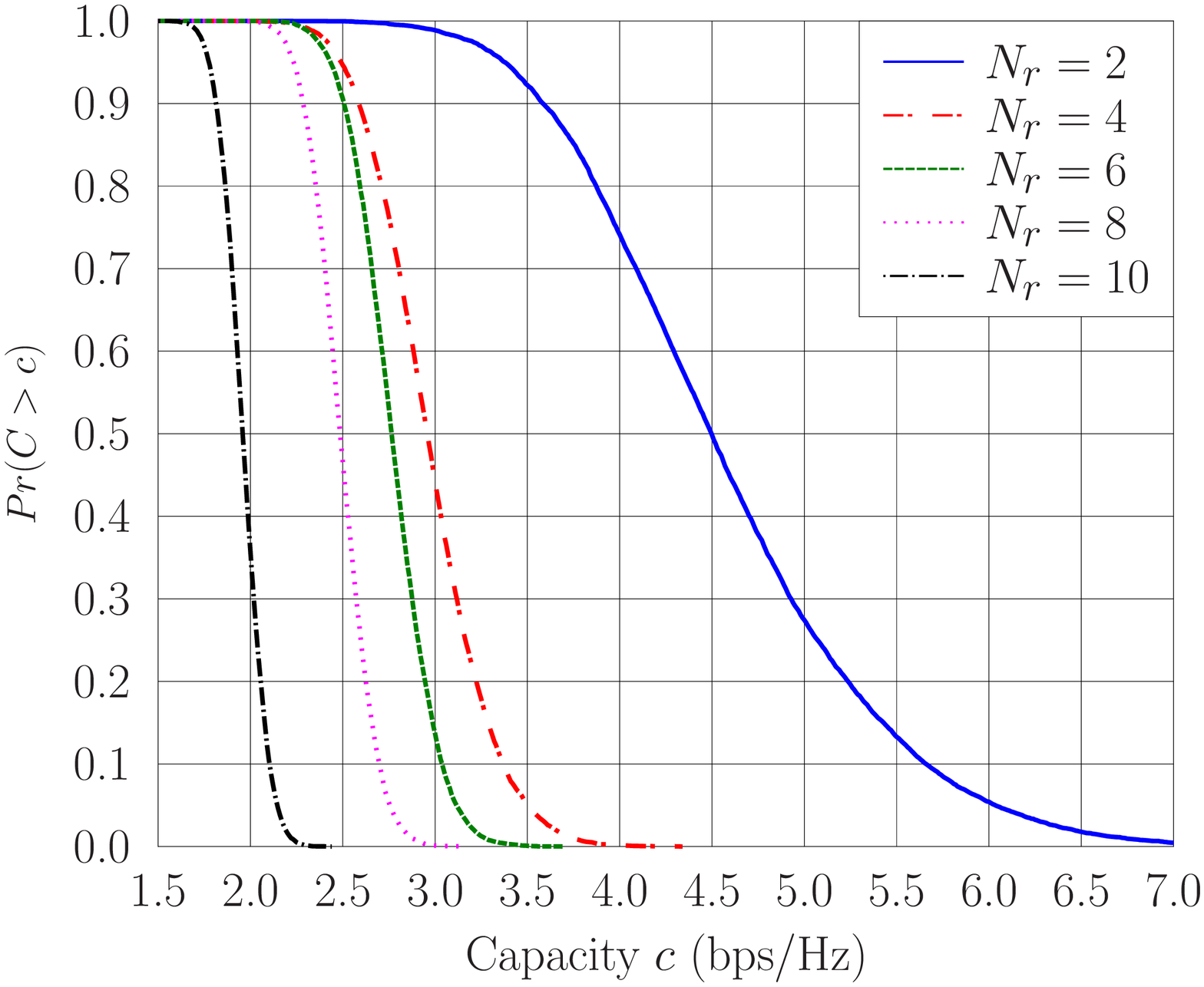}
  \label{FIG7b}}
\end{center}
\vspace{-4mm}
\caption{(a)~The achievable throughput as a function of the number of DRAs $N_r$, and
 (b)~The CCDFs of the simulated throughputs per DRA for different numbers of DRAs $N_r$.
 The distances between the interfering aircraft and the desired receiving aircraft are
 uniformly distributed within the range of $\big[d_{b^{*}}^{a^{*}}, ~ D_{\text{max}}\big]$.
 The rest of the parameters are specified in Table~\ref{Tab3}.}
\label{FIG7}
\end{figure*}

\begin{figure*}[htbp!]
\vspace{-2mm}
\begin{center}
\subfigure[Throughput]{
  \includegraphics[width=0.47\textwidth]{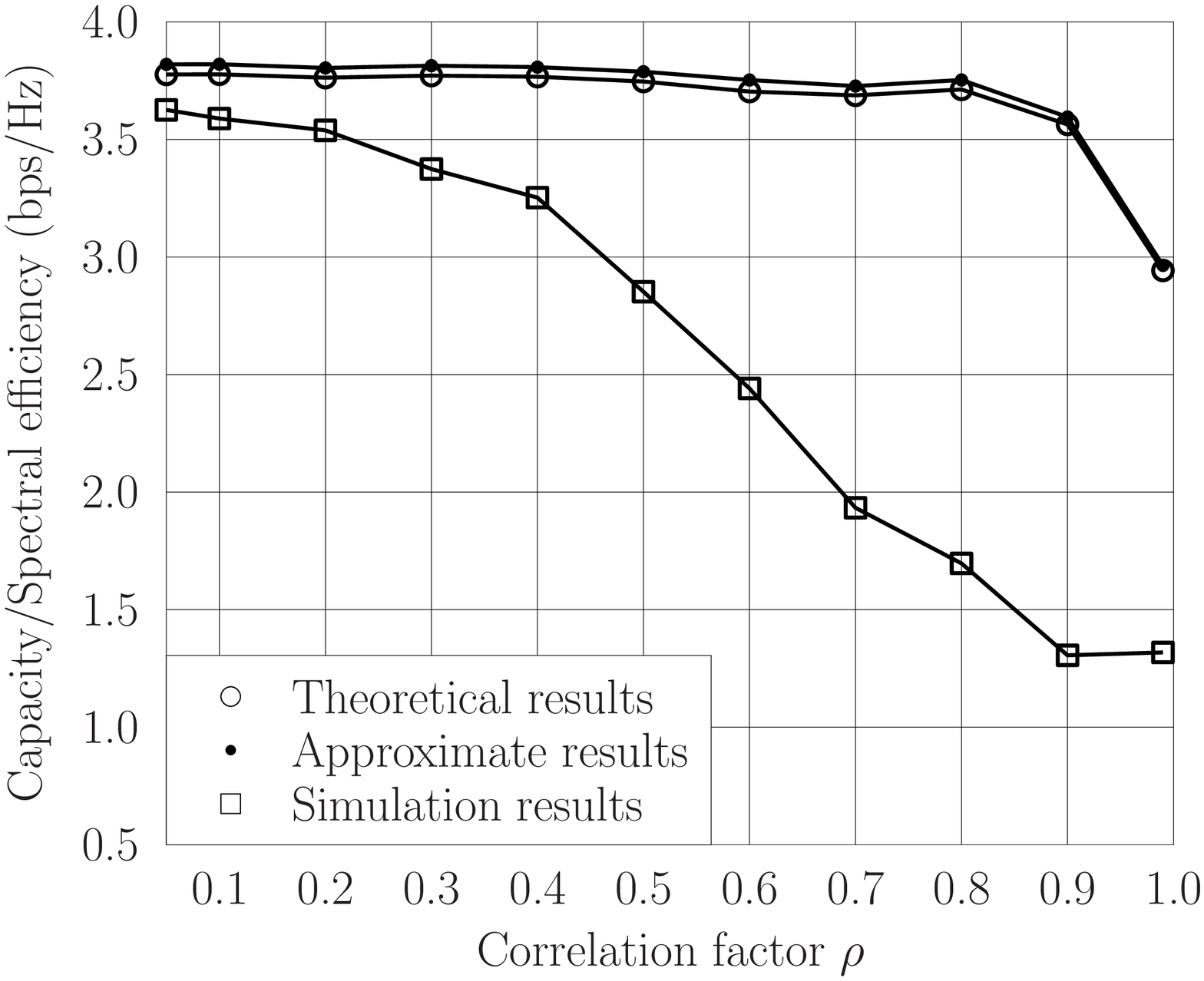} 
  \label{FIG8a}
	}
\subfigure[CCDF]{\includegraphics[width=0.47\textwidth]{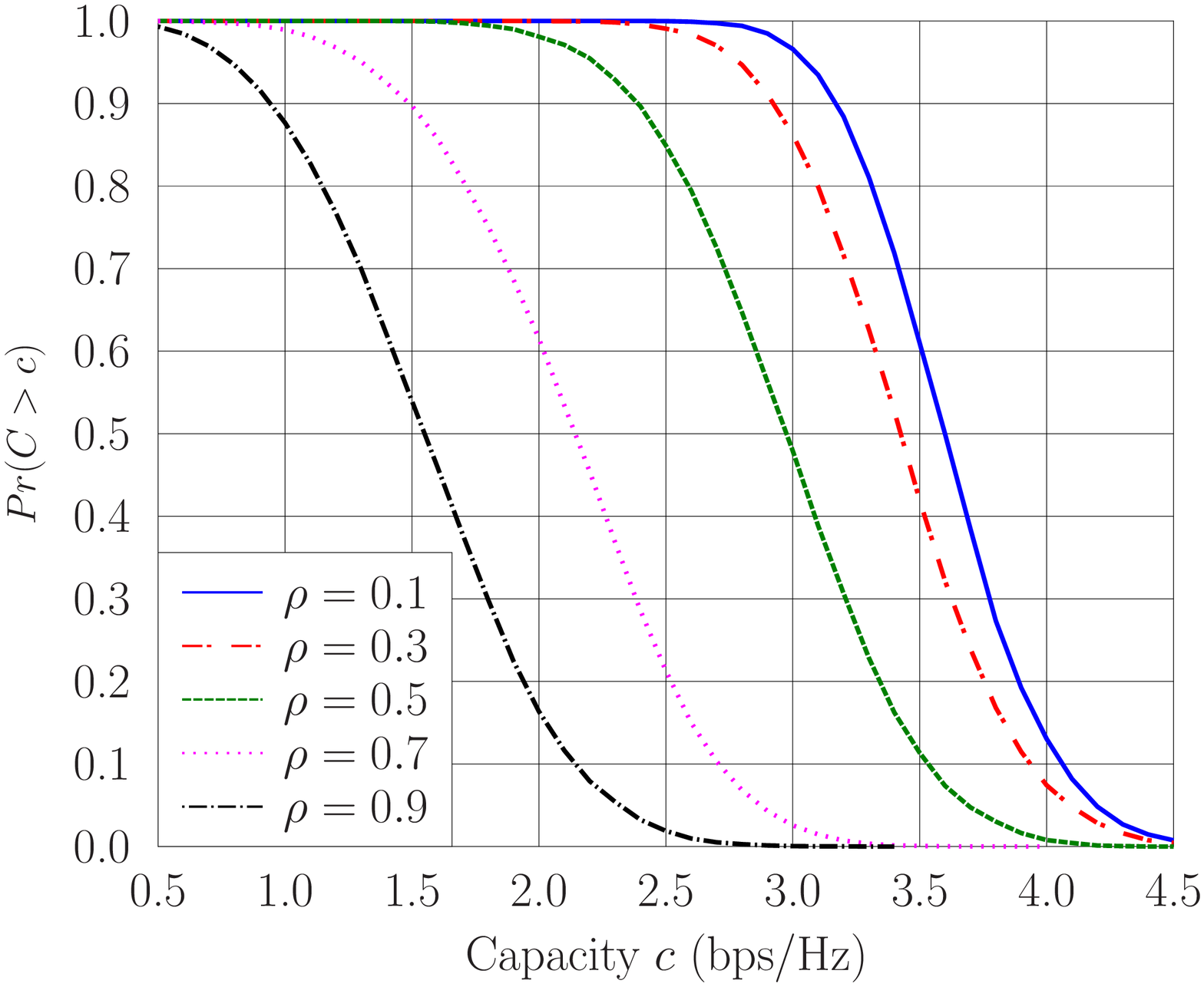}
  \label{FIG8b}}
\end{center}
\vspace{-4mm}
\caption{(a)~The achievable throughput per DRA as a function of the correlation factor of
 DTAs $\rho$, and (b)~The CCDFs of the simulated throughputs per DRA for different values
 of $\rho$. The distances between the interfering aircraft and the desired receiving
 aircraft are uniformly distributed within the range of $\big[d_{b^{*}}^{a^{*}}, ~
 D_{\text{max}}\big]$. The rest of the parameters are specified in Table~\ref{Tab3}.}
\label{FIG8}
\vspace{-1mm}
\end{figure*}

\begin{figure*}[tp!]
\vspace*{-1mm}
\begin{center}
\subfigure[Throughput]{
  \includegraphics[width=0.47\textwidth]{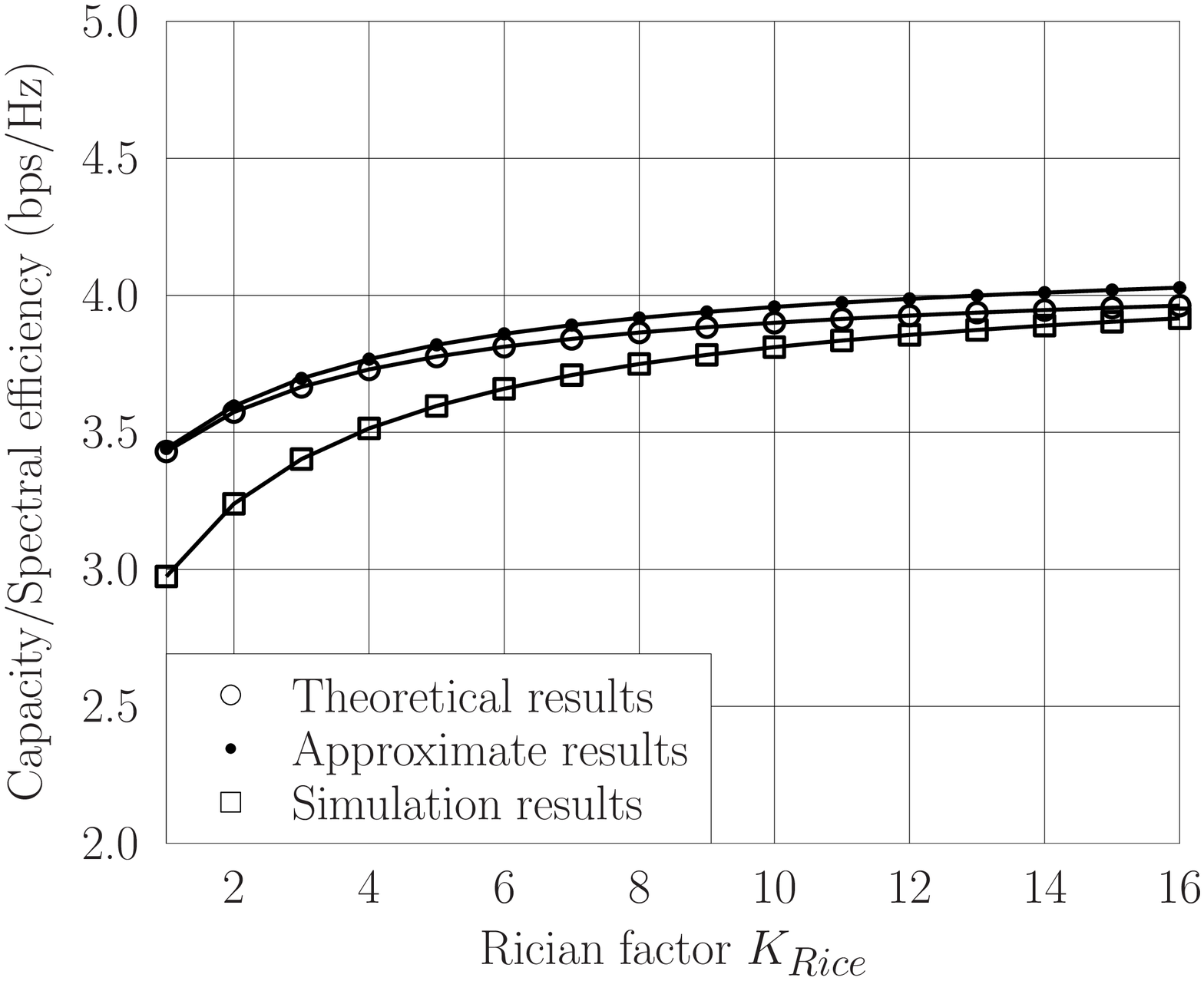} 
  \label{FIG9a}
	}
\subfigure[CCDF]{\includegraphics[width=0.47\textwidth]{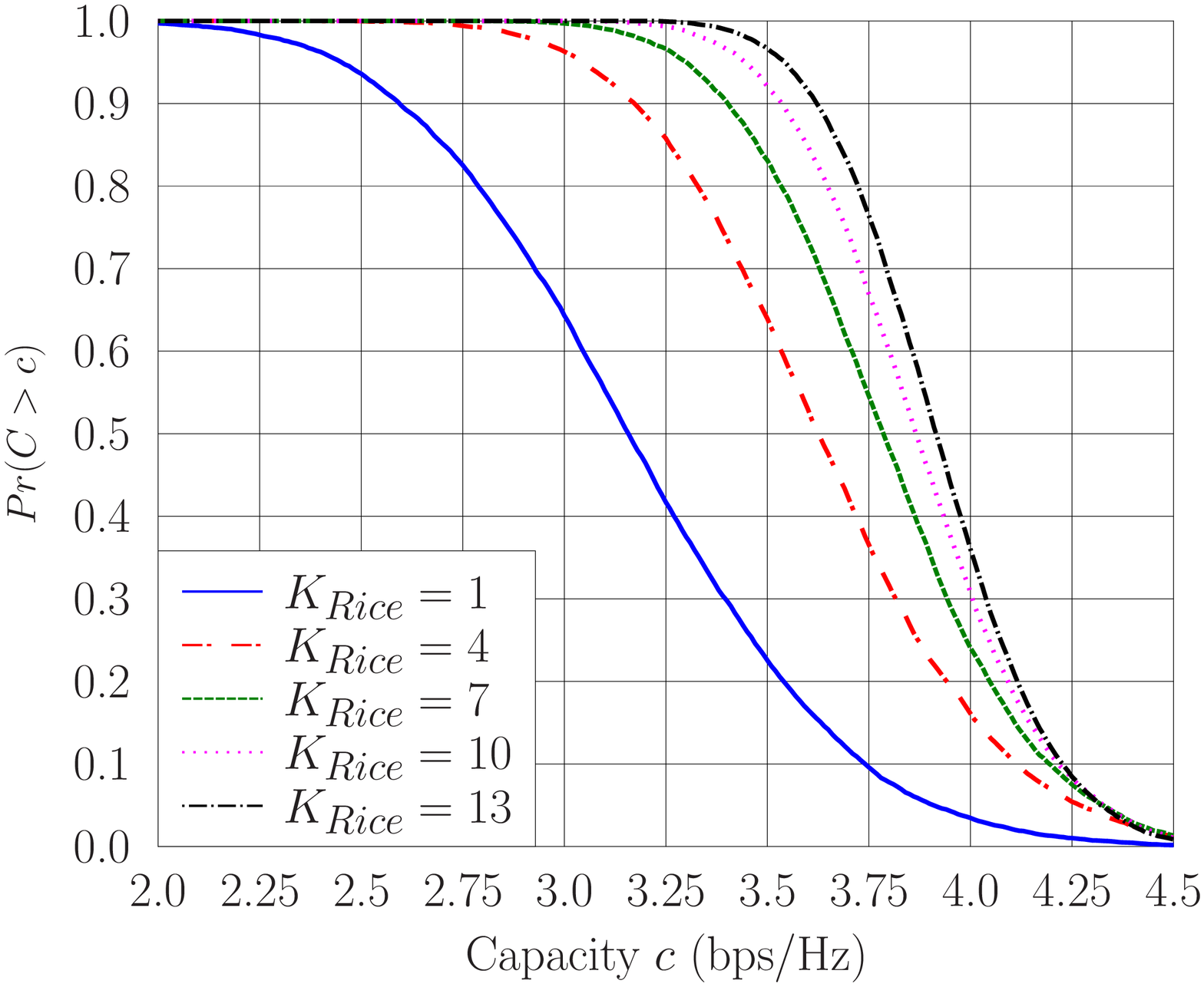}
  \label{FIG9b}}
\end{center}
\vspace{-4mm}
\caption{(a)~The achievable throughput per DRA as a function of the Rician factor
 $K_{\text{Rice}}$, and (b)~The CCDFs of the simulated throughputs per DRA for
 different values of $K_{\text{Rice}}$. The distances between the interfering aircraft
 and the desired receiving aircraft are uniformly distributed within the range of
 $\big[d_{b^{*}}^{a^{*}}, ~ D_{\text{max}}\big]$. The rest of the parameters are
 specified in Table~\ref{Tab3}.}
\label{FIG9}
\vspace{-4mm}
\end{figure*}

 Next the impact of the number of DRAs on the achievable throughput is  studied in
 Fig.~\ref{FIG7}. In particular, Fig.~\ref{FIG7a} portrays the achievable throughputs
 as a functions of $N_r$, where the left $y$-axis labels the achievable throughput
 per DRA and the right $y$-axis indicates the sum rate of the $N_r$ DRAs. As expected,
 the achievable sum rate increases with $N_r$. However, the increase in the sum rate
 is not proportional to the increase of $N_r$. In fact, it is clearly seen from
 Fig.~\ref{FIG7a} that the achievable throughput per DRA is reduced with the increase
 of $N_r$. The reason for this trend is because the inter-antenna interference increases
 with the increase of $N_r$, as seen in the third and fourth terms of (\ref{eq34}). The
 CCDFs of the simulated throughputs per DRA are illustrated in Fig.~\ref{FIG7b} for
 different $N_r$, which agree with the curve of the simulated capacity per DRA shown in
 Fig.~\ref{FIG7a}, namely the achievable throughput per DRA is lower for larger number
 of DRAs. 

 Fig.~\ref{FIG8} depicts the impact of the correlation factor $\rho$ of DTAs on the
 achievable throughput per DRA. It can be seen from Fig.~\ref{FIG8a} that strong signal
 correlation between DTAs will reduce the achievable throughput, as expected. This
 degradation is particularly serious in the Monte-Carlo simulation results, i.e. in
 practice, but less notable for the asymptotic theoretical upper bound. Note that there
 is a large gap between the theoretical upper bound and the Monte-Carlo simulation
 results when $\rho \ge 0.4$, { indicating that the asymptotic interference plus
 noise power, which is a lower bound of the true interference plus noise power, is no
 longer sufficiently tight.} The CCDFs of the simulated throughputs per DRA for
 different values of $\rho$ are shown in Fig.~\ref{FIG8b}, which statistically validates
 the curve of the simulated throughput per DRA given in Fig.~\ref{FIG8a}. 

 The impact of the Rician factor $K_{\rm Rice}$ on the achievable throughput is shown
 in Fig.~\ref{FIG9}. Specifically, Fig.~\ref{FIG9a} depicts the achievable throughputs
 per DRA as a functions of the Rician factor $K_{\rm Rice}$, while the CCDFs of the
 simulated throughput per DRA recorded for different values of  $K_{\rm Rice}$ are given
 in Fig.~\ref{FIG9b}. The results of Fig.~\ref{FIG9} clearly show that a higher Rician
 factor leads to a higher throughput. Observing the channel model (\ref{eq2}), we can
 see that a higher $K_{\rm Rice}$ results in a larger deterministic or LOS component,
 which is beneficial for the achievable performance.

\subsection{Discussions}\label{S4.2}

 In the above extensive simulation study, we have carefully investigated how the number
 of interfering aircraft $A$, the number of DTAs $N_t$, the number of DRAs $N_r$, the
 distance $d_{b^{*}}^{a^{*}}$ between the desired pair of communicating aircraft, the
 correlation factor $\rho$ between antennas, and the Rician factor $K_{\rm Rice}$ of
 the aeronautical communication channel impact on the achievable system performance of
 our distance-based ACM and large-scale antenna array aided AANET. Based on the
 simulation results, we can draw the following observations.

 The distance between the desired pair of communicating aircraft and the correlation
 factor of DTAs have adverse effects on the achievable system performance. Increasing
 $d_{b^{*}}^{a^{*}}$ and/or $\rho$ reduces the achievable transmission rate. On the
 other hand, increasing the number of DTAs, the number of DRAs, and/or the Rician
 factor is beneficial for the achievable system performance. Specifically, increasing
 $N_t$ and/or $K_{\rm Rice}$ lead to higher transmission rate, while increasing $N_r$
 also increases the total transmission rate, although the achieved throughput per DRA
 is reduced with  the increase of $N_r$.

 Most importantly, our extensive simulation results have validated the design presented
 in Section~\ref{S3} and provide the evidence that our design is capable of supporting
 the future Internet above the clouds. For example, let us consider { the AANET for airborne
 commercial Internet access} having a 5\,GHz carrier frequency and a 6\,MHz bandwidth, which
 is spatially shared by $A=14$ other aircraft { within the effective communication zone
 covered by the AANET.} When the distance between the desired pair of communicating aircraft
 is $d_{b^{*}}^{a^{*}}=10$\,km, our design is capable of offering a total data rate of
 79\,Mbps, as seen from Fig.~\ref{Fig4}. In the senario of $A=4$ and $d_{b^{*}}^{a^{*}}=70$\,km,
 our design is capable of providing a total data rate of 60\,Mbps, as observed from Fig.~\ref{FIG5}.

\section{Conclusions}\label{S5}

 A large-scale antenna array aided and novel distance-based ACM scheme has been proposed
 for aeronautical communications. Unlike the terrestrial instantaneous-SNR based ACM
 design, which is unsuitable for aeronautical communication applications, the proposed
 distance-based ACM scheme switches its coding and modulation mode according to the
 distance between the desired communicating aircraft. Based on our asymptotic closed-form
 theoretical analysis, we have explicitly derived the set of distance-thresholds for the
 proposed ACM design and have provided a theoretical upper bound of the achievable
 spectral efficiency and throughput, which has considered the impact of realistic channel
 estimation error and of co-channel  interference. Our extensive  simulation results have
 validated our design and theoretical analysis. This study therefore has provided a
 practical high-data-rate and high-spectral-efficiency solution for supporting the future
 Internet above the clouds.

\appendix

 First, we have the following lemma.\setcounter{equation}{38}
\begin{lemma}\label{L1}
 Let $\bm{A}\in \mathbb{C}^{N \times N}$ and $\bm{x}\sim \mathcal{CN}\left(\frac{1}{\sqrt{N}}\bm{m},
 \frac{1}{N}\bm{\Upsilon}\right)$, where $\frac{1}{\sqrt{N}}\bm{m}\in \mathbb{C}^{N \times 1}$ and
 $\frac{1}{N}\bm{\Upsilon}\in \mathbb{C}^{N \times N}$ are the mean vector and the covariance matrix
 of $\bm{x}$, respectively. Assume that $\bm{A}$ has a uniformly bounded spectral norm with respect
 to $N$ and $\bm{x}$ is independent of $\bm{A}$. Then we have\setcounter{equation}{38}
\begin{align}\label{eA1}
 \lim_{N\to \infty}\bm{x}^{\rm H}\bm{A}\bm{x} =& \text{Tr}\left\{\left(\frac{1}{N}\bm{M} +
 \frac{1}{N}\bm{\Upsilon}\right)\bm{A}\right\} ,
\end{align}
 where $\bm{M}=\bm{m}\bm{m}^{\rm H}$.
\end{lemma}

\begin{proof}
 Let $\bm{y}=\sqrt{N}\bm{x}-\bm{m}$. As $\bm{x}\sim \mathcal{CN}\left(\frac{1}{\sqrt{N}}\bm{m},
 \frac{1}{N}\bm{\Upsilon}\right)$, we have $\bm{y}\sim \mathcal{CN}\left(\bm{0},\bm{\Upsilon}\right)$.
 Furthermore, we have
\begin{align}\label{eA2}
 &\bm{x}^{\rm H}\bm{A}\bm{x} = \left(\frac{1}{\sqrt{N}}\bm{m} + \frac{1}{\sqrt{N}}\bm{y}\right)^{\rm H}
  \bm{A}\left(\frac{1}{\sqrt{N}}\bm{m} + \frac{1}{\sqrt{N}}\bm{y}\right) \nonumber \\
 & \hspace*{3mm}= \frac{1}{N}\bm{m}^{\rm H}\bm{A}\bm{m} + \frac{1}{N}\bm{y}^{\rm H}\bm{A}\bm{y}
  + \frac{1}{N}\bm{m}^{\rm H}\bm{A}\bm{y} + \frac{1}{N}\bm{y}^{\rm H}\bm{A}\bm{m} .
\end{align}
 Since $\bm{y}\sim \mathcal{CN}\left(\bm{0},\bm{\Upsilon}\right)$ and $\bm{y}$ does not depend on
 $\bm{m}$, according to Lemma~1 of \cite{fernandes2013inter}, we have
\begin{align}
 \lim\limits_{N \to \infty} \frac{\bm{m}^{\rm H}\bm{A}\bm{y}}{N} =& 0 , \label{eA3} \\
 \lim\limits_{N \to \infty} \frac{\bm{y}^{\rm H}\bm{A}\bm{m}}{N} =& 0 . \label{eA4}
\end{align}
 Furthermore, according to the trace lemma of \cite{hoydis2012random}, we have
\begin{align}\label{eA5}
 \lim\limits_{N \to \infty} \left(\frac{1}{\sqrt{N}}\bm{y}^{\rm H}\right)\bm{A}\left(\frac{1}{\sqrt{N}}\bm{y}\right)
   = \text{Tr}\left\{\frac{1}{N}\bm{\Upsilon}\bm{A}\right\} .
\end{align}
 Substituting (\ref{eA3}) to (\ref{eA5}) into (\ref{eA2}) leads to (\ref{eA1}).
\end{proof}

 Recalling the distribution (\ref{eq20}), we have 
\begin{align}\label{eA6}
 \mathcal{E} \left\{  \left[\widehat{\bm{H}}_{b^*}^{a^*}\right]_{[n_r^{*}:~]}^{\rm H}
  \left[\widehat{\bm{H}}_{b^*}^{a^*}\right]_{[n_r^{*}:~]} \right\} =&
  \left[\bm{\Theta}_{b^{*}}^{a^{*}}\right]_{(n_r^{*}, n_r^{*})}  .
\end{align}
 Upon setting $\bm{x}=\left[\widehat{\bm{H}}_{b^*}^{a^*}\right]_{[n_r^{*}:~]}$ in conjunction with
 $\bm{A}=\bm{I}_{N_t}$ in Lemma~\ref{L1}, we have 
\begin{align}\label{eA7}
 \lim_{N_t\to \infty} \left[\widehat{\bm{H}}_{b^*}^{a^*}\right]_{[n_r^{*}:~]}
  \left[\widehat{\bm{H}}_{b^*}^{a^*}\right]_{[n_r^{*}:~]}^{\rm H} =&
  \text{Tr}\left\{\left[\bm{\Theta}_{b^{*}}^{a^{*}}\right]_{(n_r^{*}, n_r^{*})} \right\} .
\end{align}
 Hence, for a large $N_t$, which is the case considered in this paper, we have
\begin{align}\label{eA8}
  \left[\widehat{\bm{H}}_{b^*}^{a^*}\right]_{[n_r^{*}:~]}
  \left[\widehat{\bm{H}}_{b^*}^{a^*}\right]_{[n_r^{*}:~]}^{\rm H} \approx &
  \text{Tr}\left\{\left[\bm{\Theta}_{b^{*}}^{a^{*}}\right]_{(n_r^{*}, n_r^{*})} \right\} .
\end{align}
 In addition, according to the distribution of $\left[\widetilde{\bm{H}}_{b^*}^{a^*}\right]_{[n_r^{*}:~]}$
  given in (\ref{eq25}), we have
\begin{align}\label{eA9}
 & \mathcal{E}\left\{ \left[\widetilde{\bm{H}}_{b^*}^{a^*}\right]_{[n_r^{*}:~]}^{\rm H}
  \left[\widetilde{\bm{H}}_{b^*}^{a^*}\right]_{[n_r^{*}:~]}\right\} =
  \left[\bm{\Xi}_{b^{*}}^{a^{*}}\right]_{(n_r^{*}, n_r^{*})} .
\end{align}
 
 With the aid of (\ref{eA6}) and (\ref{eA8}) as well as (\ref{eA9}), we can readily derive
 (\ref{eq30}), as shown in the following
\begin{align}\label{eA10}
 & \text{Var}\left\{\left[\bm{H}_{b^*}^{a^*}\right]_{[n_r^{*}:~]}
  \left[\bm{V}_{b^*}^{a^*}\right]_{[~:n_r^{*}]}\right\} \nonumber \\
 &
  \hspace*{3mm}= \mathcal{E}\Bigg\{\Bigg| \left(\left[\widehat{\bm{H}}_{b^*}^{a^*}\right]_{[n_r^*:~]}
  + \left[\widetilde{\bm{H}}_{b^*}^{a^*}\right]_{[n_r^*:~]}
  \right)\left[\widehat{\bm{H}}_{b^*}^{a^*}\right]_{[n_r^{*}:~]}^{\rm H} 
 - \text{Tr}\left\{\left[\bm{\Theta}_{b^{*}}^{a^{*}}\right]_{(n_r^{*}, n_r^{*})}
  \right\}\Bigg|^2\Bigg\} \nonumber \\
 & \hspace*{3mm}\approx \mathcal{E}\Bigg\{\Bigg| \left[\widetilde{\bm{H}}_{b^*}^{a^*}\right]_{[n_r^*:~]}
  \left[\widehat{\bm{H}}_{b^*}^{a^*}\right]_{[n_r^{*}:~]}^{\rm H}\Bigg|^2\Bigg\} 
  = \mathcal{E}\Bigg\{\! \left[\widetilde{\bm{H}}_{b^*}^{a^*}\right]_{[n_r^{*}:~]}
  \left[\widehat{\bm{H}}_{b^*}^{a^*}\right]_{[n_r^{*}:~]}^{\rm H}
  \left[\widehat{\bm{H}}_{b^*}^{a^*}\right]_{[n_r^{*}:~]}
  \left[\widetilde{\bm{H}}_{b^*}^{a^*}\right]_{[n_r^{*}:~]}^{\rm H}\! \Bigg\} \nonumber \\
 & \hspace*{3mm}= \mathcal{E}\Bigg\{ \left[\widetilde{\bm{H}}_{b^*}^{a^*}\right]_{[n_r^{*}:~]}
  \mathcal{E}\bigg\{ \left[\widehat{\bm{H}}_{b^*}^{a^*}\right]_{[n_r^{*}:~]}^{\rm H}
  \left[\widehat{\bm{H}}_{b^*}^{a^*}\right]_{[n_r^{*}:~]} \bigg\} 
  \left[\widetilde{\bm{H}}_{b^*}^{a^*}\right]_{[n_r^{*}:~]}^{\rm H} \Bigg\}\nonumber \\
 &
  = \mathcal{E}\left\{\left[\widetilde{\bm{H}}_{b^*}^{a^*}\right]_{[n_r^{*}:~]}
  \left[\bm{\Theta}_{b^{*}}^{a^{*}}\right]_{(n_r^{*}, n_r^{*})}
  \left[\widetilde{\bm{H}}_{b^*}^{a^*}\right]_{[n_r^{*}:~]}^{\rm H}\right\} \nonumber \\
 & \hspace*{3mm} = \text{Tr}\left\{\left[\bm{\Xi}_{b^{*}}^{a^{*}}\right]_{(n_r^{*}, n_r^{*})}
  \left[\bm{\Theta}_{b^{*}}^{a^{*}}\right]_{(n_r^{*}, n_r^{*})} \right\} .
\end{align}


\end{document}